\documentclass[11pt]{article}
\usepackage{sigsam}
\usepackage{amsmath,amsthm,amsfonts,amscd,amssymb}

\usepackage{enumitem}
\usepackage{faktor}
\usepackage{mathrsfs}
\usepackage{kantlipsum}
\usepackage{mathtools}
\usepackage{bm}

\usepackage{multicol}
\usepackage[linesnumbered,vlined,boxruled,boxed,algo2e]{algorithm2e}
\usepackage{algorithm}
\usepackage{algorithmic}
\usepackage{tikz}
\usepackage{xcolor}
\usepackage{url}

\newtheorem{theorem}{Theorem}[section]

\newtheorem{definition}{Definition}

\newtheorem{lemma}[theorem]{Lemma}
\newtheorem{proposition}[theorem]{Proposition}
\newtheorem{problem}[theorem]{Problem}
\newtheorem{assumption}[theorem]{Assumption}
\newtheorem{claim}[theorem]{Claim}
\newtheorem{remark}{Remark}
\newtheorem{example}{Example}

\newcommand{\CC}{\mathbb{C}}
\newcommand{\ZZ}{\mathbb{Z}}

\newcommand{\NN}{\mathbb{N}}

\newcommand{\VV}{\mathbb{V}}
\newcommand{\PP}{\mathbb{P}}
\newcommand{\Aff}{\mathbb{A}}

\newcommand{\CZ}{\mathcal{CZ}}
\newcommand{\CF}{\mathcal{CF}}

\newcommand{\cR}{\mathcal{R}}

\newcommand{\x}{\ensuremath{\boldsymbol{x}}}

\renewcommand{\u}{\ensuremath{\boldsymbol{u}}}

\newcommand{\f}{\ensuremath{\boldsymbol{f}}}
\newcommand{\p}{\ensuremath{\boldsymbol{p}}}
\newcommand{\ff}{\ensuremath{\boldsymbol{f}}}

\newcommand{\balpha}{\ensuremath{\boldsymbol{\alpha}}}
\newcommand{\bbeta}{\ensuremath{\boldsymbol{\beta}}}

\newcommand{\OO}{\ensuremath{\mathcal{O}}\xspace}
\newcommand{\sO}{\ensuremath{\widetilde{\mathcal{O}}}\xspace}
\newcommand{\OB}{\ensuremath{\mathcal{O}_B}\xspace}
\newcommand{\sOB}{\ensuremath{\widetilde{\mathcal{O}}_B}\xspace}

\newcommand{\bsz}{\ensuremath{\mathtt{h}}\xspace}

\newcommand{\M}{\ensuremath{\mathtt{M}}\xspace}

\DeclareMathOperator{\rk}{rk}

\DeclareMathOperator{\SL}{SL}

\DeclareMathOperator{\supp}{supp}
\DeclareMathOperator{\codim}{codim}
\DeclareMathOperator{\Gr}{Gr}
\DeclareMathOperator{\mdeg}{mdeg}
\DeclarePairedDelimiter\abs{\lvert}{\rvert}

\issue{Vol. 57, No. 4, Issue 226, December 2023}
\articlehead{}
\titlehead{On the Complexity of Chow and Hurwitz Forms}
\authorhead{M. Levent Dogan, Alperen A. Erg\"{u}r and Elias Tsigaridas}
\begin{document}

\title{On the Complexity of Chow and Hurwitz Forms}

\author{
M. Levent Dogan \\
Technische Universit\"{a}t Berlin \\
Berlin, Germany\\
\url{dogan@math.tu-berlin.de}
\and
Alperen A. Erg\"{u}r\\
The University of Texas at San Antonio\\
San Antonio, Texas, U.S.A.\\
\url{alperen.ergur@utsa.edu}
\and
Elias Tsigaridas\\
INRIA Paris and Sorbonne Universit\'{e}\\
Paris, France\\
\url{elias.tsigaridas@inria.fr}
}

\date{}

\maketitle

\begin{abstract}
We consider the bit complexity of computing Chow forms of projective varieties defined over 
integers and their generalization to multiprojective spaces. We develop a deterministic algorithm 
using resultants and obtain a single exponential complexity upper bound. Earlier computational 
results for Chow forms were in the arithmetic complexity model; thus, our result 
represents the first bit complexity bound. We also extend our algorithm to Hurwitz forms in 
projective space and we explore connections between multiprojective Hurwitz forms and matroid 
theory. The motivation for our work comes from incidence geometry where intriguing computational 
algebra problems remain open.
\end{abstract}

\section{Introduction}
Suppose a curve in space is given geometrically, e.g., in parametric form, how can one find an algebraic representation of this curve? This was a question tackled by Cayley \cite{cayley}  and later generalized to arbitrary varieties by van der Waerden and Chow by introducing what is now called the Chow form \cite{Waerdenvander1937}. The Chow form is now recognized as a  fundamental construction in algebraic geometry and it is particularly important in elimination theory. The structural and computational aspects of Chow forms have been an active area of research for several decades,  we humbly provide a sample of references: \cite{kapranov-bernd-chow,dalbec2,caniglia1990compute,jeronimo2004computational,kohn_coisotropic_2021}. Despite the large body of literature on the subject, one basic aspect has received little attention: the bit complexity of computing a Chow form.  Our paper fills this gap. We also extend our algorithmic results to Hurwitz forms \cite{sturmfels_hurwitz_2017} and to the recent generalization of Chow forms for multiprojective varieties \cite{multigraded}.

Another motivation for deriving precise complexity bounds for computing Chow forms comes from combinatorics: Let  $S_1 , S_2 \subset \mathbb{C}^{2}$  be two finite sets and let $p$ be a $4$-variate polynomial. How many zeros of $p$ can be located in $S_1 \times S_2$? It was noticed in \cite{dezeeuw_schwartz} that this simple question has surprising consequences in extremal combinatorics and incidence geometry. This question almost entirely looks like a subject for the Schwartz-Zippel-De Millo-Lipton (SZDL) lemma \cite{lipton}, but $S_1$ and $S_2$ are two dimensional.  

In \cite{dogan_multivariate} we developed a  multivariate generalization of the SZDL lemma: Suppose $\lambda = (\lambda_1 , \lambda_2, \ldots, \lambda_m)$ is a partition of $n$, that is $n =\sum_{i=1}^m \lambda_i$.  Let $0\neq p\in\CC[x_1,x_2,\dots,x_n]$ be a polynomial of degree $d$ and assume that for any collection of  positive dimensional varieties $V_i \subset \mathbb{C}^{\lambda_i}$, for $i=1,2,\ldots,m$, we have $V_1 \times V_2 \times \ldots \times V_m \not\subset \VV(p)$. Then, for any collection of finite sets $S_i \subset \mathbb{C}^{\lambda_i}$, any real $\varepsilon >0$, and $S := S_1 \times S_2 \times \ldots \times S_{m}$ we have 
\[ \abs{\VV(p) \cap S} = O_{n,d,\varepsilon}\Big( \prod_{i=1}^m \abs{S_i}^{1-\frac{1}{\lambda_i+1}+\varepsilon} + \sum_{i=1}^m \prod_{j \neq i} \abs{S_j} \Big) . \]
Note that the containment assumption on the variety $\VV(p)$ is necessary for any non-trivial upper bound to hold: one can simply place any collection of finite sets $S_i$ with arbitrary size on the positive dimensional varieties $V_i$. For applications in incidence geometry, we need to certify this assumption on the polynomial $p$ that encodes the incidence relation. This brings us to the following problem.


\begin{problem} \label{badinclusion}
Assume that we are given  an $n$-variate polynomial $p$ and a positive integer vector $\lambda=(\lambda_1, \lambda_2,\ldots,\lambda_m)$  where $n=\sum_{i=1}^m \lambda_i$. Decide whether there exist positive dimensional varieties $V_i \subset \mathbb{C}^{\lambda_i}$ for $i=1,2,3,\ldots,m$ such that 
$
 V_1\times V_2\times\dots\times V_m \subset \VV(p) . 
 $ 
\end{problem}
Note that a variety $\VV(p)$ of degree $d$ can contain varieties of arbitrary degree (think of a high degree curve included in an hyperplane). Hence, the problem resists standard computational algebra tools that have to assume a degree bound on $V_i$. Our paper \cite{dogan_multivariate} includes an algorithm to decide if $\VV(p)$ contains a cartesian product of hypersurfaces. One can hope to utilize multiprojective Chow form(s) of $\VV(p)$ to relate the general containment Problem~\ref{badinclusion} to the special case of hypersurface containment.  This was our motivation to develop precise complexity bounds for computing multiprojective Chow forms.
\subsection{Previous Works and Our Results}
The Chow form of a variety can be computed by using standard tools of
elimination theory, e.g., Gr\"obner basis, in a black-box manner \cite{dalbec2}.
This black-box approach does not exploit the special structure of the problem
and does not yield precise complexity estimates. Instead, we will rely on
resultant computations. Roughly speaking, for a (homogeneous) polynomial system
of $n$ equations in $m \geq n$ variables, the resultant is a polynomial in
$m -n$ variables (and in the coefficients of the original polynomials) that is
zero if and only if the original system has a solution;
we refer to \cite{gkz,cox} for further details. Resultants are typically computed using a formula that expresses them as a factor of the determinant of a square matrix. Other factors of this determinant can sometimes be identically zero and might prevent us to compute the resultant. Canny~\cite{canny_generalised_1990} introduced the generalized characteristic polynomial that symbolically perturbs the input polynomials and avoid the unsolicited vanishing of components. Our computations rely on Canny's technique.

To our knowledge, the first algorithm with a precise complexity estimate to
compute the Chow form of a pure dimensional variety, say $V$, is due to Caniglia
\cite{caniglia1990compute}. Caniglia's algorithm is based on a clever reduction
to linear algebra and it admits a single exponential upper bound on the number
of arithmetic operations. In the case where the defining polynomials of $V$ are
given by straight-line programs, Jeronimo et al.
\cite{jeronimo2004computational} describe a probabilistic algorithm that
computes the Chow form of $V$, see also \cite{jeronimo2001computing}. This Las
Vegas algorithm admits a single exponential time upper bound on a
Blum-Shub-Smale (BSS) machine. Its expected complexity is polynomial in terms of
the input size and the geometric degree of the variety $V$ and thus; a single
exponential time worst case complexity. See
\cite[Theorem~1]{jeronimo2004computational} for the exact statement.

\textbf{Our contribution} There is an extensive literature on the complexity of elimination theory procedures in general, and complexity of polynomial system solving in particular; we provide a small sample here \cite{bitpiss1,bitpiss2,bitpiss3,bitpiss4,bitpiss5}. Despite the strong literature on the subject, we were not able to locate any results on the bit complexity of computing Chow forms. We present a single exponential time algorithm to
compute the Chow form of a pure dimensional variety where the computational
model is the bit model; see Proposition~\ref{prop:CV-homo-CI} for the complete
intersection case and Theorem~\ref{thm:CV-homo} for the general case. We
further extend our algorithmic techniques to compute Hurwitz forms
\cite{sturmfels_hurwitz_2017} with precise complexity estimates, see Proposition~\ref{lem:Hurwitz-complexity}.

There is a generalization by Osserman and Trager \cite{multigraded} of Chow forms for varieties in multiprojective space. Our method to compute the Chow form is based on resultant computations
and this seamlessly extends to multiprojective space, see Lemma~\ref{lem:CV-mhomo-CI} for the complete intersection case and Theorem~\ref{thm:MultiChowForm} for the general case.
We should emphasize that the generalization from the projective Chow forms to the multiprojective ones is far from straightforward both from the mathematical and  algorithmic complexity point of views. Even though a multiprojective space is isomorphic to a projective variety via the Segre embedding, this requires adding many more additional variables. To work directly
with multiprojective spaces and avoid the use of (many more) additional variables
we have to take into account the partition of the variables in blocks,
the combinatorics of the supports of the polynomials,
and exploit this structure both from a mathematical and an algorithmic point of view.
Moreover, the number of blocks (and the variables in each block) should also
appear in the corresponding complexity estimates.
We refer to Section~\ref{sec:Chow-mhomo} for a detailed presentation.

In addition, we discuss a multihomogeneous generalization of the Hurwitz form. To the best of our knowledge, our paper provides the first result in this area. Contrary to the homogeneous case, multigraded Chow forms and Hurwitz form require a choice of a non-degenerate multidimension vector for the linear subspace, in a sense that is discussed in Section~\ref{sec:Chow-mhomo}. This set of non-degenerate dimension vectors gives rise to an interesting combinatorial structure, namely a \textit{polymatroid}. In \cite{clz}, it has been proven that the set of multidegrees of a multiprojective variety forms a polymatroid. In \cite{multigraded}, the authors show that the set of non-degenerate dimension vectors for Chow forms equals to the \textit{truncation} of this polymatroid, which is itself a polymatroid. In a similar fashion, we show that non-degenerate dimension vectors for Hurwitz forms also form a polymatroid, which is obtained by the \textit{elongation} of the polymatroid of Chow forms. We discuss this combinatorial structure in Section~\ref{sec:fun}.

Last but not least, our techniques allow us to provide precise bit complexity estimates
for the coefficients of the Chow form (both in the projective and in the multiprojective case).
These bounds give an estimation on the size of the objects that we compute with
and consist a measure of hardness of the corresponding algorithmic problems.
\subsection{Outline of the paper} The rest of the paper is structured as follows.
In section \ref{prelim} we present a short overview of Chow forms. Section~\ref{sec:Chow-homo} is on the computation of the Chow form in $\PP^n$ and the extension of techniques to compute Hurwitz forms.  Section~\ref{sec:Chow-mhomo} presents algorithms for computing multiprojective Chow forms. Section \ref{sec:fun} explores connections between multiprojective Chow \& Hurwitz forms and  matroid theory.

\section{Preliminaries} \label{prelim}

\subsection{Notation} The bold small letters indicate vectors or points; in particular $\x = (x_0, \dots, x_n)$ or $\x = (x_1, \dots, x_n)$ depending on the context. We denote by $\OO$, resp. $\OB$, the arithmetic, resp.  bit, complexity and we use $\sO$, resp. $\sOB$, to ignore (poly-)logarithmic factors. For a polynomial $f \in \ZZ[\x]$, $\bsz(f)$ denotes the maximum bitsize of its coefficients; we also call it the bitsize of $f$. We use $[n]$ to denote the set $\{1, 2, \dots,n \}$. Throughout $\CC$ denotes the field of complex numbers, $\ZZ$  integers, $\Aff^n$ the   \textit{affine space}, and $\PP^n$ the \textit{projective space}. 

Given polynomials $\ff \coloneqq ( f_1,f_2,\dots,f_k ) \in\CC[\x]^{k}$, we call the zero locus
\[
\VV_{\Aff}(f_1,f_2,\dots,f_k)\coloneqq\{\p \in\Aff^n\mid f_1(\p)=f_2(\p)=\dots=f_k(\p)=0\},
\] 
 the \textit{affine variety} defined by $\ff$. In the case that $\ff$ consists of homogeneous polynomials, the set \[
 \VV_{\PP}(f_1,f_2,\dots,f_k)\coloneqq\{[\p]\in\PP^n\mid f_1(\p)=f_2(\p)=\dots=f_k(\p)=0\}
 \] is well-defined and called the \textit{projective variety} defined by $\f$. 
 
 A projective variety $V\subset\PP^n$ is called \textit{irreducible} if we cannot write it as a non-trivial union of two subvarieties. Otherwise, $V$ is called \textit{reducible}.
 We can write every reducible variety (in an essentially unique way) as a finite union
 of irreducible subvarieties, that is
\begin{equation}
\label{eq:decomposition}
V = \bigcup\nolimits_{i=1}^l V_i,\quad\quad 
(V_i\subsetneq \bigcup\nolimits_{j\neq i} V_j) .
\end{equation} 
The irreducible subvarieties $V_i$ are the \textit{irreducible components} (or
components for short) and the expression (\ref{eq:decomposition}) is the
\textit{irreducible decomposition} of $V$.

As a preparation for the discussion on Chow forms, we define the \textit{dimension} of a projective variety \textit{\`{a} la Harris} \cite[\S 11]{harris}: $\dim V$ of $V$ is the integer $r$ satisfying the property that \textit{every} linear subspace $L\subset\PP^n$ of dimension at least $(n-r)$ intersects $V$ and a \textit{generic} subspace $L\subset\PP^n$ of dimension at most $(n-r-1)$ is disjoint from $V$.
We call $V$  \textit{pure dimensional} (sometimes also \textit{equidimensional}), if every irreducible component of $V$ has the same dimension. Note that if $V$ is irreducible, then it is trivially pure dimensional.

\subsection{Associated hypersurfaces and Chow forms} 
The main object of our study is the associated hypersurface of a variety $V$ and its defining polynomial, the Chow form. For a detailed introduction on Chow forms, we refer to \cite{dalbec,dalbec2}. For a concise, clear, and deeper exposition we refer to \cite{gkz}.

Let $V\subset\PP^n$ be an irreducible variety of dimension $r$. 
By the definiton of $\dim V$, if $L\subset\PP^n$ is a generic linear subspace
of dimension $n-r-1$, then the intersection $L\cap V$ is empty. The \textit{associated hypersurface} of $V$ is the set of non-generic subspaces, i.e.,
$(n-r-1)$-dimensional subspaces of $\PP^n$ that have a non-empty intersection with $V$.
To be more concrete, we consider the \emph{Grassmannian}
\[
\Gr(n-r-1,n) = \{L\subset\PP^n\mid L\text{ is a subspace of dimension }n-r-1\},
\] 
of linear subspaces of $\PP^n$ of dimension $n-r-1$.

\begin{proposition}
\label{prop:ass}
	Let $V\subset\PP^n$ be an irreducible variety of dimension $r$. 
	Then, the set of linear subspaces intersecting $V$,
	\[
	\mathcal{CZ}_V \coloneqq \{L\in \Gr(n-r-1,n)\mid V\cap L\neq\varnothing\} \subseteq Gr(n-r-1,n), 
	\] 
 	is an irreducible hypersurface of $\Gr(n-r-1,n)$ that we  call the \textbf{associated hypersurface} of $V$. Moreover, $\mathcal{CZ}_V$ uniquely defines $V$; that is,
 	\[
	V = \{\p \in\PP^n\mid \p\in L\,\text{ implies }\, L\in\mathcal{CZ}_V\}.
	\]
\end{proposition}

It is known that $\Gr(n-r-1,n)$ has the property that every hypersurface of $\Gr(n-r-1,n)$ is the zero locus of a single element of its coordinate ring (see, for example, \cite[Proposition~2.1]{gkz}). In particular, the associated hypersurface $\mathcal{CZ}_V$ is the zero set of an element in the coordinate ring\footnote{For a general variety $X$, it is not true that every hypersurface in $X$ is the zero locus of some element of the coordinate ring of $X$. The standard example is a plane curve of degree $d>2$ and a point on the curve.} of $\Gr(n-r-1,n)$ that we call the \textit{Chow form} of $V$. 
We write a linear subspace $L\in\Gr(n-r-1,n)$ 
 as the intersection of $r+1$ hyperplanes. 
 For $0\neq \u\in \CC^{n+1}$, we consider $U(\u,\x) = u_0 x_0+\dots+u_n x_n$.
\begin{definition}
\label{def:Chowform}
Let $V\subset\PP^n$ be a variety. The \textit{Chow form} of $V$ is the square-free polynomial with the property that \[
\CF_V(\u_0,\dots,\u_r) = 0\quad\iff \quad V \cap \VV(U(\u_0,\x),\dots,U(\u_r,\x))\neq\varnothing
\] for $\u_0, \u_1,\dots, \u_r\in\CC^{n+1}$. The Chow form is defined only up to multiplication with a non-zero scalar.
\end{definition}
The previous definition and Proposition~\ref{prop:ass} imply that if $V$ is irreducible, then $\CF_V$ is an irreducible polynomial which defines the irreducible hypersurface $\CZ_{V}$. More generally, \[
V= V_1 \cup V_2 \cup\dots\cup V_l \quad \Rightarrow\quad \CF_{V} = \CF_{V_1}\times \CF_{V_2}\times\dots\times \CF_{V_l}
\] holds if $V$ is pure dimensional and $V=\bigcup_{i=1}^l V_i$ is the irreducible decomposition of $V$. If $V$ is not pure dimensional, then a linear subspace of complementary dimension generically does not intersect the lower dimensional components so the Chow form forgets these components. We will generally assume that $V$ is pure dimensional.

\begin{lemma}
\label{lem:gcd}
	Suppose $V\subset\PP^n$ is a pure dimensional variety such that 
	$V=X_1\cap X_2\cap \dots\cap X_l$, 
	where the variety $X_i$ is pure dimensional and 
	$\dim V=\dim X_i$, for all $i \in [l]$. Then \[
	\CF_V = \gcd( \CF_{X_1}, \CF_{X_2},\dots, \CF_{X_l}).
	\]
\end{lemma}
\begin{proof}
	Since $\dim V=\dim X_i$ for $i \in [l]$, we can see that the irreducible components of $V$ are exactly the common irreducible components of $X_i$. Hence, both sides are equal to the product of the Chow forms of the components of $V$ and this finishes the proof.
\end{proof}

\section{The Chow Form in $\PP^n$}
\label{sec:Chow-homo}

In what follows  $V\subset\PP^n$ denotes an equidimensional 
projective variety of dimension $r$. We present an algorithm 
to compute the Chow form, $\mathcal{CF}_V$, of $V$.

\subsection{The case of a complete intersection} 
\label{sec:Proj-complete}

\begin{algorithm}[h]      
	\caption{\xspace{ \textsc{ChowForm\_CI} }} 
	\label{alg:Chow-homo-CI}
	
	\begin{flushleft}
	\textbf{Input:} $f_1, \dots, f_{n-r} \in \ZZ[\x].$  \\
	\textbf{Precondition: } $V=\VV(f_1,\dots,f_{n-r})$ is pure $r$-dimensional.\\
	\textbf{Output:} The Chow form of $V$.
	\end{flushleft}
	\begin{enumerate}
    	\item Consider $r+1$ linear forms,  
         $U_i = \sum_{j=0}^{n} u_{ij}x_j, \text{ for } 0 \leq i \leq r $ .\;
        \item Eliminate the variables $\x_i$
          \[
          	R = \mathtt{Elim}(\{f_1, \dots, f_{n-r}, U_0, \dots, U_{r}\}, \{\x \}) \in \ZZ[u_{i,j}] .
          \]
	    \item $R_r = \textsc{SquareFreePart}(R)$.
        \item (Optional) Apply the straightening algorithm.
        \item \textsc{return} $R_r$.
    \end{enumerate}
\end{algorithm}

Assume that $V\subset\PP^n$ is a \textit{set theoretic complete intersection} over $\ZZ$, i.e., $V$ is the common zero locus of $\codim(V)=n-r$ many polynomials \[
V = \VV(f_1,f_2,\dots,f_{n-r})\subset \PP^n,
\]
where $f_i \in \ZZ[\x]$ and $\deg(f_i) = d_i$. Let $d\coloneqq\max_i d_i$ denote the maximum of the degrees. Moreover, assume $\bsz(f_i) \leq \tau$, i.e., the bitsizes of $f_i$ are all bounded by $\tau$.

Let $\bm{u} = \{ u_{ij}\, : \, 0 \leq i \leq r,\, 0 \leq j \leq n \}$ 
be a set of formal variables and consider the system of $r+1$ formal linear combinations \[
U(\u, \x) := \begin{cases}
 \; U_0 :=  u_{00}x_0+u_{01}x_1+\dots+u_{0n}x_n = 0\\
 \; U_1 := u_{10}x_0+u_{11}x_1+\dots+u_{1n}x_n = 0\\
 \quad \quad \quad \quad \quad \quad \quad\vdots\\
 \; U_r := u_{r0}x_0+u_{r1}x_1+\dots+u_{rn}x_n = 0.
 \end{cases}
\] 
We will regard each polynomial in  $U$, $U_i$, as a polynomial in $\CC[\u][\x]$.
Suppose for a specialization of the variables $u: u_{ij}\in\CC,\, 0\leq i\leq r,\, 0\leq j\leq n$
the matrix that they naturally correspond to, say $M_{u} \in\CC^{(r+1) \times (n+1)}$,
is of full rank. Consequently
\[
L = \VV(U(u,\mathbf{x}))\subset\PP^n
\] 
is an $(n-r-1)$-dimensional linear subspace of $\PP^n$. Conversely, every $(n-r-1)$-dimensional 
subspace of $\PP^n$ is of this form;
that is, it induces a full rank $(r+1) \times (n+1)$ matrix. Moreover,  $V\cap L\neq\varnothing$ holds if and only if the overdetermined system
\begin{equation}
\label{eq:system}
U(\u,\x) = 0,\; f_1(\x)=0,\; \dots,\; f_{n-r}(\x)=0
 \end{equation} 
 has a  solution in $\PP^n$.

\begin{proposition}
	\label{prop:Chow-res}
  Let $\mathcal{R}\in\CC[\u]$ be the resultant of the system  of $n+1$ polynomials
  \[
    U(\u, \x)=0, f_1(\x)=0,\dots,f_{n-r}(\x) = 0,
  \]
  eliminating the variables $x_0,x_1,\dots,x_n$. Then, \begin{enumerate}
 	\item $\mathcal{R}$ is invariant under the action of $\SL_{r+1}$ on $\CC[\u]$ via the left multiplication, and, 
 	\item $\mathcal{R}(u)=0$ if and only if either (the corresponding matrix) $M_{u}$ has rank $<r+1$, or, $M_{u}$ is full rank and $\VV(U(u,\x))\in\CZ_V$.
 \end{enumerate}
\end{proposition}
\begin{proof}
  The invariance of the resultant under matrices of determinant $1$ is a standard result, see, for example, \cite[Section~3.3]{cox}. For the second claim, note that the system 
  (\ref{eq:system}) has a solution if and only if $L=\VV(U(u,\x))$ intersects
  $V$. As $\dim L=n-\rk(u)$, $L$ intersects $V$ precisely when $M_u$ is
  rank-deficient, or, $M_u$ is full-rank and $L\in\CZ_V$.
\end{proof}

\begin{proposition}
	\label{prop:Chow-redundant}
Let $\mathcal{R}=\mathcal{R}_1^{e_1}\mathcal{R}_2^{e_2}\dots\mathcal{R}_c^{e_c}$ be the irreducible factorization of $\mathcal{R}$. Then $c$ equals the number of irreducible components of $V$ and the \textit{square-free part} 
$\mathcal{R}_1\mathcal{R}_2\dots \mathcal{R}_c$
of $\mathcal{R}$ is the Chow form of $V$. 
\end{proposition}
\begin{proof}
    By Proposition~\ref{prop:Chow-res}, the zero locus of $\mathcal{R}$ coincides with the zero locus of the Chow form, $\CF_V$ (Definition~\ref{def:Chowform}). Thus, the square-free part of $\mathcal{R}$ equals $\CF_V$ and the number of irreducible factors of $\mathcal{R}$ and $\CF_V$ are equal.
\end{proof}

\begin{proposition}
  \label{prop:CV-homo-CI}
  Consider $I = \langle f_{1}, \dots, f_{n-r} \rangle \subseteq \ZZ[x_{0}, \dots, x_{n}]$
  where each $f_{i}$ is homogeneous of degree at most  $ d$ and has bitsize $\tau$;
  also the corresponding projective variety, $V$, has pure dimension $r$.
  Let $\u_i := (u_{i,0}, \dots, u_{i, n})$,
  for $0 \leq i \leq r$, be $(n+1)(r+1)$ new variables.
  The Chow form of $V$, $\CF_V$, is a multihomogeneous polynomial;
  it is
  homogeneous in each block of variables $\u_{i}$
  of degree at most $d^{n-r}$
  and has bitsize $\sO(n d^{r-1}\tau)$.
  The algorithm
  {\normalfont \textsc{ChowForm\_CI}} (Alg.~\ref{alg:Chow-homo-CI}) correctly computes
  $\CF_V$ in
  \[
    \sOB( n^{n^{2} + \omega n} d^{2(n-r)(r+1)(n+1) + (\omega+1)n+ r - 1} (n + d + \tau))
  \]
bit operations, where $\omega$ is the exponent of matrix multiplication.
\end{proposition}

\begin{proof}
  The correctness of the algorithm follows from Propositions~\ref{prop:Chow-res} and \ref{prop:Chow-redundant}.

To compute the Chow form, following Alg.~\ref{alg:Chow-homo-CI}, 
we introduce $r+1$ linear forms, say
\[
  U_{i} = u_{i,0} x_{0} + \cdots + u_{i,n} x_{n} ,
\]
where  $\u_i := (u_{i,0}, \dots, u_{i, n})$, 
for  $0 \leq i \leq r$, are $(n+1)(r+1)$ new variables.
The Chow form is the square-free part of  
the resultant of the system 
\[ \bm{F} = \{f_{1}, \dots, f_{n-r}, U_{0}, \dots, U_{r}\} , \]
say $\mathcal{R}$,
where we consider the polynomials as elements in $(\ZZ[\u])[\x]$
and we eliminate the variables $\x$; thus $\CF_V \in \ZZ[\u_{0}, \dots, \u_{r}]$.

\medskip

\noindent
\emph{Bounds on the degree and the bitsize.}
We compute the resultant $\cR$ using the Macaulay matrix, $M$, corresponding to
the system $\bm{F}$, as quotient of two determinants, that is
$\cR = \det(M) / \det(M_{1})$, where $M_{1}$ is a submatrix of $M$.
To avoid the case where the denominator is zero, that is $\det(M_{1}) = 0$, we
apply generalized characteristic polynomial technique of
Canny~\cite{canny_generalised_1990}. For this we symbolically perturb the
polynomials using a new variable $s$; that is
$\hat{f}_{i} = f_{i} + s x_{i}^{d_i}$, for $i \in [n-r]$.
Now the system becomes
\[
  \hat{\bm{F}} := \{ \hat{f}_{1}, \dots, \hat{f}_{n-r}, U_{0}, U_{1}, \dots, U_{r} \},
\]
where we consider the elements in $\hat{\bm{F}}$  as polynomials in the variables $\x$
with coefficients in $\ZZ[\bm{u}, s]$.

The resultant of the new system, say $\hat \cR$, that we obtain after eliminating
the variables $\x$, is a polynomial in $\ZZ[\bm{u},s]$. We recover $\cR$ as the
first non-vanishing coefficient of $\hat \cR$, by interpreting the latter as a
univariate polynomial in $s$.
The resultant is a multihomogeneous polynomial in the coefficients of the input polynomials \cite[Chapter~3]{cox}. In our case, the monomials of  $\hat{\cR}$ are  of the form
\[
  \varrho \, \bm{a}_{1}^{\mathcal{W}_{1}}\cdots \bm{a}_{n-r}^{\mathcal{W}_{n-r}}
  \bm{b}_{0}^{\mathcal{W}_{n-r+1}} \cdots \bm{b}_{r}^{\mathcal{W}_{n}},
\]
where $\rho \in \ZZ$. The integer $\rho$, roughly speaking,
corresponds to the lattice points of the Newton polytopes of the input
polynomials, we refer to \cite{sombra2004height} for further details.
The interpretation of $\bm{a}_{i}$ is that it represents a product
of coefficients of $\hat{f}_{i}$ of total degree $W_{i}=|\mathcal{W}_i|$, for
$i \in [n-r]$.  Similarly, $\bm{b}_{j}^{\mathcal{W}_{n-r+j}}$ represents a
product of coefficients of $U_{j}$ of total degree $W_{n-r + j}=|\mathcal{W}_{n-r+j}|$, for
$0 \leq j \leq r$.  Finally, $W_{k}$ is the B\'ezout bound on the
solutions of the system $\hat{\bm{F}}$ if we exclude the $k$-th
equation.
It holds $W_{i} \leq d^{n-r-1}$, for $i \in [n-r]$,
and $W_{n-r+j} \leq d^{n-r}$, for $0 \leq j \leq r$.

The degree of $\hat{\cR}$ w.r.t. $s$ is at most $(n-r)\max\{W_1,W_2,\dots,W_{n-r}\}$. Its coefficients, and so
also $\cR$, by interpreting $\hat{\cR}$ as a univariate polynomial in $s$, are
multihomogeneous polynomials w.r.t. each block of variables $\u_{i}$ of degree bounded by
$W_{n} \leq d^{n-r}$.

To bound the bitsize of $\hat{\cR}$, and thus the bitsize of $\cR$, we
follow the same techniques as in \cite[Theorem~5]{emt-dmm-j}.
For a  worst case bound, it suffices to consider
that each $\bm{a}_i$ is of the form $(s + 2^\tau)$.
Thus, following Claim~\ref{claim:mul-mpoly},
the bitsize of $\bm{a}_i^{\M_i}$
is at most $\OO((n-r)d^{n - r -1} \tau + (n-r)^2 d^{n-r-1}\lg(n d))$ which is $\sO((n-r)d^{n-r-1} (\tau + n - r))$.
Hence, the product of all of them has bitsize
$ \OO((n-r)^2d^{n - r -1} \tau + (n-r)^3 d^{n-r-1}\lg(n d))
  = \sO(n^2d^{n-r-1} (n + \tau)) $.
Similarly, each $\bm{b}_j^{\M_{n-r+j}}$ has bitsize
$\OO((n-r)d^{n-r}\lg{d})$
and the product of all of them
$\OO( (n-r)rd^{n-r}\lg{d} + n (n-r)^2 r \lg(rd))
= \sO(n^2 (d^{n-r} + n))$. In addition, it holds that $\bsz(\rho) = \sO(n^2 d^{n-r})$ \cite[Table~1]{emt-dmm-j}. Putting all the bounds together, we deduce that
the bitsize of $\hat{\cR}$, and hence the bitsize of $\cR$ is
\[
  \sO(n^2 d^{n-r-1}(n + d + \tau)) .
\]
\medskip
\noindent \textbf{Computing the determinant(s).} To actually compute the resultant we exploit Kronecker's trick
and efficient algorithms for computing the determinant of matrices with polynomial entries.

Let $D_1 := (n-r)d^{n-r-1}$ be a bound on the degree of $s$ and
$D_2 := d^{n-r}$ a bound on the variables $\bm{u}$
in the polynomial $\hat{\cR}$. We perform the following substitutions
\[
  u_{0,0} \to s^{(D_1 +1)}, u_{0,1} \to s^{(D_1+1)(D_2+1)}, u_{0,2} \to s^{(D_1+1)(D_2+1)^2},
  \dots,
u_{r,n} \to s^{(D_1+1)(D_2+1)^{(n+1)(r+1) - 1}}.
\]
In this way the elements of the Macaulay matrix $M$ become univariate polynomials in $s$
of degree at most
$(D_{1}+1)(D_{2}+1)^{(n+1)(r+1)}$
and bitsize $\tau$.
Also $M$ is a square $m \times m$ matrix,
where $m$ corresponds to  the number of homogeneous monomials of degree
$\sum_{i=1}^{n-r}(d-1) + \sum_{i=0}^{r}(1-1) + 1  = (n-r) (d-1)  + 1 $ in $n+1$ variables;
that is  $m  =\binom{(n-r)(d-1) + 1 + n}{n} \leq ((n-r)d)^{n}$.

Now we compute the quotient $\det(M(s)) / \det(M_1(s)) \in \ZZ[s]$
and we recover $\cR$ from the first non-vanishing coefficient of this polynomial.
The computation of the determinants
costs at most $\sO(m^{\omega} \, (D_{1}+1)(D_{2}+1)^{(n+1)(r+1)})$ operations
and if we multiply by the bitsize of the output, then we deduce
that the computation of the resultant $\cR$
costs
\[
  \sOB( 2^{(n+1)(r+1)} n^{n^{2} + \omega n} d^{[(n-r)(r+1)+1](n-r) + (\omega+1)n+ r - 1} (n + d + \tau))
\]
bit operations,
where $\omega$ is the exponent of the complexity of matrix multiplication. \\
\textbf{Square-free factorization.} Finally, we have to compute the square-free part of $\mathcal{R}$.
This amounts, roughly, to one $\gcd$ computation.
For polynomials in $\nu$ variables, of degree $\delta$ and bitsize $L$, 
the gcd costs $\sOB(\delta^{2\nu}L)$ \cite[Lemma~4]{lpr-asympt-17}.
This translates to
\[\sOB(n^{2} (r+1)^{2(n+1)(r+1)} d^{2(n-r)(n+1)(r+1)+n-r-1}(n+d+\tau)).\] Combining the two complexity bounds, after some simplifications, we obtain the claimed result.
%
%
\end{proof}

\subsection{The case of an over-determined system}
\label{sec:Proj-over}

\begin{algorithm}[h]      
	\caption{\xspace{ \textsc{ChowForm} }} 
	\label{alg:Chow-homo}
	
	\begin{flushleft}
	\textbf{Input:} $f_1, \dots, f_{m} \in \ZZ[\x], r\in\NN$  \\
	\textbf{Precondition:} Assumption~\ref{ass:degree}.\\
	\textbf{Output:} The Chow form of $\VV(f_1, \dots, f_m)$.
	\end{flushleft}
	\begin{enumerate}
	\item $\Lambda^1,\dots,\Lambda^N \coloneqq \textsc{GenericLC}(f_1,f_2,\dots,f_m)$.
    	\item \lFor{$r \in [N]$}{$F_i = \textsc{ChowForm\_CI}(\Lambda_{\ff}^i)$}
    	      
        \item \textsc{return} $\gcd( F_1, \dots, F_N )$ 
    \end{enumerate}
\end{algorithm}

\noindent
Now we remove the assumption of complete intersection. 
Consider  
\[
V = \VV(f_1,f_2,\dots,f_m)\subset\PP^n,
\]
where $m\geq n-r=\codim V$. Moreover, if $d_i\coloneqq\deg(f_i)$ 
we further assume that $d_1\geq d_2\geq\dots\geq d_m$.
As before, we want to add to our system $r+1$ linear forms 
and eliminate the variables $\x$. However, if we simply add 
the linear forms, then we end up with more than $n+1$ polynomials. 
Thus, we cannot use resultant computations, at least directly, 
to perform elimination. 

To overcome this obstacle we consider the following observation.
\begin{proposition}
    Every pure-dimensional variety $V\subset\PP^n$ can be written as the intersection of finitely many complete intersections: $V = V_1 \cap V_2\cap\dots\cap V_l$ where $\forall i,\, \dim V_i=\dim V$.
\end{proposition}
The proposition offers a strategy to compute $\CF_V$: 1) Compute complete intersections $V_1,V_2,\dots,V_l$ such that $V$ equals their intersection, 2) compute the Chow form $\CF_{V_i}$ of each $X_i$ by the means of Algorithm~\ref{alg:Chow-homo-CI}, 3) compute the gcd of $\CF_{V_i}, i=1,2,\dots,l$. Since $V$ is the intersection of $X_i$, we have $\CF_V=\gcd(\CF_{V_i}\mid i=1,2,\dots,l)$ by Lemma~\ref{lem:gcd}.

In order to compute complete intersections $V_i$ whose intersection is $V$, we will proceed as follows: We replace the original system $\ff$ 
with a generic system of $\codim(V) =n-r$ many polynomials $\tilde{f}_i$ 
that vanish on $V$, by choosing $\tilde{f}_i$ to be generic linear combinations of $f_i$.
We will prove that the zero locus of the 
new system is a pure $r$-dimensional variety that contains $V$
(Proposition~\ref{prop:genericlambda}). 
By repeating this process, say $k$ times, we obtain a number of pure dimensional varieties, $V_1,V_2,\dots,V_k$, all containing $V$. For large enough $k$, 
the intersection $V_1\cap\dots\cap V_k$ is exactly $V$. What the exact number of required pure-dimensional varieties itself is an interesting question. Proposition~\ref{prop:number} gives the upper bound $k=\lceil \frac{m}{n-r}\rceil$.
The Chow form of $V$ satisfies 
\[
\CF_V = \gcd(\CF_{V_1}, \dots, \CF_{V_k}).
\] 
Moreover, each $V_i$ is a set theoretic complete intersection
and so we can use Alg.~\ref{alg:Chow-homo-CI} to compute its Chow form $\CF_{V_i}$.



First, we modify the set of polynomials $\ff$ so that it contains only
polynomials of the same degree. Let $d=\max_i d_i$. 
We replace each $f_i$ satisfying $d_i<d$ with the set of polynomials
	 \[
x_0^{d-\deg f_i} f_i,\; x_1^{d-\deg f_i}f_i,\;\dots,\; x_n^{d-\deg f_i}f_i.
\] 
The zero locus of the new system, which has less than $(n+1)m$ polynomials, 
equals the zero locus of the original system, but now the polynomials
all have the same degree.
So in what follows we make the following assumption:
\begin{assumption}
\label{ass:degree}
$V=\VV(f_1,f_2,\dots,f_m)\subset\PP^n$ is a pure dimensional variety of dimension $\dim(V)=r$, 
where $f_1,f_2,\dots,f_m$ are homogeneous polynomials of the same degree $d$. 
\end{assumption}

The assumption that $f_i$ all have the same degree allows us to consider linear combinations of $f_i$.
That is, for $\lambda_1,\lambda_2,\dots,\lambda_m\in\CC$, the polynomial 
$\sum_{i=1}^m \lambda_i f_i$ is also a homogeneous polynomial of degree $d$, 
and, in particular, it defines a projective hypersurface.
More generally, for an arbitrary $k$ and a matrix $\Lambda = [\lambda_{ij}]\in \CC^{k \times m}$ we define the system 
\begin{equation}
\label{eq:lambdasystem}
\Lambda_{\ff}:=\begin{cases}
	\lambda_{11}f_1+\lambda_{12}f_2+\dots+\lambda_{1m}f_m\\
	\lambda_{21}f_1+\lambda_{22}f_2+\dots+\lambda_{2m}f_m\\
	\qquad\qquad\qquad\quad\vdots\\
	\lambda_{k1}f_1+\lambda_{k2}f_2+\dots+\lambda_{km}f_m,
	\end{cases}
\end{equation}
that consists of $k$ linear combinations of $f_i$'s. 
Let $\VV(\Lambda_{\ff})$ denote the zero locus of this system. The next proposition shows that for generic $\Lambda\in\CC^{k\times m}$, the variety $\VV(\Lambda_{\ff})$ is of the form \[
\VV(\Lambda_{\ff}) = V \cup X,
\] where $X$ is a pure dimensional variety of dimension $n-k$. The proof  follows, mutatis mutandis, \cite[Section~3.4.1]{giusti} which considers the case $k=n$. 
\begin{proposition}
\label{prop:genericlambda}
For a generic choice of $\Lambda\in\CC^{k\times m}$, the components of
$\VV(\Lambda_{\ff})$ are either the components of $V$ or of dimension $n-k$.
More concretely, there exists a hypersurface
$H\subset\CC^{k\times m}$ of degree at most $k d^{k-1}$ such that the condition holds for any $\Lambda\in\CC^{k\times m}\setminus H$.
\end{proposition}
\begin{proof}
We will proceed by induction on $k$. For $k=1$, the condition is violated if and only if 
$\Lambda_{\ff} = \sum_{i=1}^{m} \lambda_i f_i \equiv  0$.
This is a linear condition on $\lambda_i$'s.
To see this, consider the matrix that has the (coefficients of the) polynomials $f_i$
as rows. Then it suffices to require $\Lambda$ not belong to the left kernel of this matrix. Let $H_1$ be an arbitrary hyperplane (i.e., hypersurface of degree $1$) containing the left kernel. Then any $\Lambda\not\in H_1$ satisfies the condition and this proves the base case. 

Let $k>1$ and assume that the claim holds for $k-1$. Let
$\Lambda\in\CC^{(k-1)\times m}\setminus H_{k-1}$ so
$\VV(\Lambda_{\ff})=V\cup X$ for some pure $(n-k+1)$-dimensional variety $X$.
Let
\[
  X = X_1\cup X_2\cup\dots\cup X_c
\]
be the irreducible decomposition of $X$ and disregard the components that are
fully contained in $V$. By the B{\'e}zout bound, we know that the number of
irreducible components are at most  $c\leq d^{k-1}$. Now we pick arbitrary
points
\[
  x_i\in X_i\setminus V,
\]
so for each $i$ there exists $j$ with $f_j(x_i)\neq 0$, and form the matrix
\[
M=\begin{bmatrix}
f_1(x_1) & f_2(x_1) & \dots & f_m(x_1)\\
f_1(x_2) & \ddots & & \vdots\\
\vdots & & \ddots & \vdots\\
f_1(x_c) & \dots & \dots & f_m(x_c)
\end{bmatrix}.
\]
Suppose $\bm{\mu} \in\CC^m$ is a vector such that each entry of $M \bm{\mu} \in\CC^c$ is non-zero.
Then the linear combination $ \tilde{f} := \mu_1 f_1+\mu_2 f_2+\dots+\mu_m f_m$
satisfies $\tilde{f}(x_i)\neq 0$ for all $i=1,2,\dots,c$. In particular we have
$X_i\not\subset\VV(\tilde{f})$. By Krull's principal ideal theorem, (see, for
example, \cite[Theorem~10.1]{eisenbud}) the intersection
$\VV(\tilde{\f})\cap X_i$ is either empty or pure dimensional of dimension
$\dim X_i-1$. Hence, the system $\Lambda_{\ff}$ together with $\tilde{f}$
satisfies the assertion.

The condition that each entry of $M\bm{\mu}\in\CC^c$ being non-zero amounts to
$\bm{\mu}$ avoiding $c$ (not necessarily distinct) hyperplanes, hence a
hypersurface $H'$ of degree at most $c\leq d^{k-1}$. In particular, the pairs
$(\Lambda,\bm{\mu})$ with $Z(\Lambda_{\ff},\tilde{\ff})$ not satisfying the
condition are contained in the
hypersurface \[ \big( H_{k-1}\,\times\,\CC^m\big)\, \cup\, \big( \CC^{(k-1)\times m}\,\times\, H'\big),
\] which has degree $(k-1)d^{k-2}+d^{k-1}\leq k d^{k-1}$ by induction.
\end{proof}



We apply the procedure in Proposition~\ref{prop:genericlambda} with $k=n-r$ and obtain a pure $r$ dimensional variety $\VV(\Lambda_{\ff})$ that contains $V$. We use \textsc{ChowForm\_CI} (Alg.~\ref{alg:Chow-homo-CI}) to compute its Chow form.
If $V$ is not a \textit{set theoretic complete intersection}, then  $\VV(\Lambda_{\ff})\neq V$ (since $\VV(\Lambda_{\ff})$ is a set theoretic complete intersection by its construction) so $\VV(\Lambda_{\ff})$ contains $V$ properly. In this case, by repeating the process of Proposition~\ref{prop:genericlambda} sufficiently many times we can construct varieties  $\VV(\Lambda^i_{\ff})$, each being a set theoretic complete intersection, where $V$ equals to their intersection. The next proposition implies that we only need $\lceil\frac{m}{n-r}\rceil$ many complete intersections.
\begin{proposition}
\label{prop:number}
Let $V=\VV(f_1,f_2,\dots,f_m)$ be as in Assumption~\ref{ass:degree} and
$N = \lceil \frac{m}{n-r}\rceil$. For a generic choice of
$\Lambda^1, \Lambda^2, \dots, \Lambda^N \in \CC^{(n-r)\times m},$ the
corresponding varieties $\VV(\Lambda^i_{\ff})$ are pure dimensional varieties of
dimension $r$ and $V = \bigcap_{i=1}^N \VV(\Lambda^i_{\ff})$. More concretely,
there is a hypersurface $H\subset\CC^{N(n-r)\times m}$ of degree at most
$ N (n-r)d^{n-r-1}+m=\OO(md^{n-r-1})$ such that for any
$(\Lambda^1,\dots,\Lambda^N)\in \CC^{N(n-r)\times m}\setminus H$, the condition
is satisfied.
\end{proposition}
\begin{proof}
    Consider the matrix 
    \[
    \Xi = [\Lambda^1,  \Lambda^2,  \cdots,  \Lambda^N]^{\top}
    \in\CC^{N(n-r)\times m},
    \] 
    and let $\VV(\Xi_{\ff})=\bigcap_{i=1}^N \VV(\Lambda^i_{\ff}).$ For generic choices of $\Lambda^i$, the matrix $\Xi$ has full rank. Since $N=\lceil \frac{m}{n-r}\rceil$, we have $N(n-r)\geq m$ so $\Xi$ is injective. Thus, we have \[
    \langle\Lambda^1_{\ff},\Lambda^2_{\ff},\dots,\Lambda^N_{\ff}\rangle = \langle\ff\rangle
    \] which implies that 
    $V = \bigcap_{i=1}^N \VV(\Lambda^i_{\ff})$. 
    
    Note that $\Xi$ satisfies the condition if and only if each $\Lambda^i$ avoids the hypersurface of degree $(n-r)d^{n-r-1}$ from Proposition~\ref{prop:genericlambda} and $\Xi$ is full rank. We can guarantee the second condition by enforcing a particular maximal minor of $\Xi$ to be non-zero. Thus, $\Xi$ satisfies the condition if it avoids $N$ hypersurfaces of degree $(n-r)d^{n-r-1}$ and a hypersurface of degree $m$.
\end{proof}

\begin{algorithm}[h]      
	\caption{\xspace{ \textsc{GenericLC} }} 
	\label{alg:GenericLambda}
	
	\begin{flushleft}
	\textbf{Input:} $f_1, \dots, f_{m} \in \ZZ[\x], r\in\NN$  \\
	\textbf{Precondition: } Assumption~\ref{ass:degree}.\\
 	\textbf{Output:} $(\Lambda^1,\Lambda^2,\dots,\Lambda^N)$.\\
 	\textbf{Postcondition:} See Proposition~\ref{prop:number}.
	\end{flushleft}
	\begin{enumerate}
    	\item $N\coloneqq\lceil \frac{m}{n-r}\rceil$.
        \item $S\coloneqq [N(n-r)d^{n-r-1}+m+1]\subset\NN$
    	\item \lFor{$(\Lambda^1,\Lambda^2,\dots,\Lambda^N)\in S^{N(n-r)\times m}$}{
    	\\ \textbf{if} $\dim(\VV(\Lambda^i_{\ff}))\leq r$ and $\Xi$ is full-rank \textbf{then}
        \\\textsc{return} $(\Lambda^1,\dots,\Lambda^N)$}
    \end{enumerate}
\end{algorithm}

\begin{lemma}
\label{lem:glc-alg}
Algorithm \ref{alg:GenericLambda} returns $N$ matrices $\Lambda^1,\Lambda^2,\dots,\Lambda^N\in\CC^{(n-r)\times m}$ satisfying the requirements of Proposition~\ref{prop:number} in $\tau m^{2m^2+\OO(1)}(2d)^{m^2n+\OO(n)}$ bit operations.
\end{lemma}

\begin{proof}
  The matrix
  $\Xi=(\Lambda^1,\dots,\Lambda^N)$ satisfies the requirements of
  Proposition~\ref{prop:number} if and only if it avoids a hypersurface
  $H\subset\CC^{N(n-r)\times m}$ of degree $\leq D=N(n-r)d^{n-r-1}+m$. By the
  bound on its degree, $H$ cannot contain a grid
  $S^{N(n-r)m}\subset\CC^{N(n-r)\times m}$ where $S\subset\CC$ is a finite set
  of size $|S|>D$. By going through all $|S^{N(n-r)m}|\leq (2md^{n-r-1})^{2m^2}$ points of the grid and testing
  membership to $H$ at each step, we can generate $\Xi$ satisfying the
  requirement.

  The membership test to $H$ amounts to checking if (i) $\VV(\Lambda^i_{\ff})$
  has dimension $\leq r$ and (ii) $\Xi$ has rank $m$. If $S$ is chosen to be the
  list of first $D+1$ natural numbers, then the entries of $\Xi$ have bitsizes
  $\log D+1=\OO(\log m+(n-r)\log d)$, so the polynomials in $\Lambda^i_{\ff}$
  have bitsizes bounded by $\OO(\tau+\log m+(n-r)\log d)$. Hence, whether
  $\dim\VV(\Lambda^i_{\ff})\leq r$ can be tested in $\tau m^{\OO(1)}d^{\OO(n)}$
  (see \cite{koiran_randomized_1997,chistov}). Whether $\Xi$ is full-rank can be
  tested in $\OO(\tau m^{\OO(1)} n \log d)$. We repeat this process
  $(2md^{n-r-1})^{2m^2}$ times, so the total complexity becomes $\tau m^{m^2+\OO(1)} (2d)^{m^2n+\OO(n)}$.
\end{proof}

\begin{theorem}
	\label{thm:CV-homo}
	Consider the ideal $I = \langle f_{1}, \dots, f_{m} \rangle \subseteq \ZZ[x_{0}, \dots, x_{n}]$,
	where each $f_{i}$ is homogeneous of degree $d$ and bitsize $\tau$;
	also the corresponding projective variety, $V$, has pure dimension $r$.
	Let $\u_i := (u_{i,0}, \dots, u_{i, n})$, 
	for $i \in [r + 1]$, be $(n+1)(r+1)$ new variables.
	The Chow form of $V$, $\CF_V$, is a multihomogeneous polynomial;
	it is 
	homogeneous in each block of variables $\u_{i}$
	of degree $d^{r}$
	and has bitsize $\sO(n d^{r-1}\tau)$.
	{\normalfont \textsc{ChowForm}} (Alg.~\ref{alg:Chow-homo}) computes $\CF_V$ in 
		\[
		\sOB(m^{2m^2+\kappa} \,n \, r^{6nr} (n-r)^{(\omega+1)n}\, (2 d)^{2m^2n+\omega n^2 r + (\omega+1)n)} \, (\tau+n)),
		\]
	bit operations where $\omega$ is the exponent of matrix multiplication and $\kappa$ is a small constant, depending on the precise complexity estimate of the dimension test in Alg.~\ref{alg:GenericLambda}.
\end{theorem}
\begin{proof}
	The correctness of the algorithm follows from the 
	previous discussion and Proposition~\ref{prop:number}.
	
	We apply Alg.~\ref{alg:GenericLambda} to generate $\Lambda^1,\dots,\Lambda^N$ such that they fulfill the assumptions of Proposition~\ref{prop:number}. The cost of this algorithm is $\tau m^{2m^2+\kappa}(2d)^{m^2n}$.
	
	The bitsize of the polynomials in $\Lambda^i_{\ff}$ is $\OO(\tau+\log m+n\log d)=\sO(\tau+n)$. Hence, we can compute the Chow form of each $\VV(\Lambda^i_{\ff})$ using Alg.~\ref{alg:Chow-homo-CI} within the complexity 
	\[
	\sOB(n(n-r)^{(\omega+1)n} r^{6nr} (2d)^{\omega n^2r+(\omega+1)n}(\tau+n)).
	\] We multiply by the number of systems, $N=\OO(m)$ to conclude.
	
	Finally, we compute the gcd of $N=\OO(m)$ Chow forms. As each Chow form has $(r+1)(n+1)$-variables, bitsize $\sO(nd^{r-1}(\tau+n))$ and degree $(r+1)d^r$, this operation costs $\sOB(m n r^2 d^{3r}(\tau+n))$, which is less than the claimed cost. 
\end{proof}

\begin{remark}
We have assumed in Alg.~\ref{alg:Chow-homo} that $r=\dim V$ is part of the input. We could also compute $r$ using the algorithms in \cite{koiran_randomized_1997,chistov}, without changing the single exponential nature of the complexity of the algorithm.
\end{remark}
\subsubsection{Straightening Algorithm}
Let $\CC[\u]$ denote the ring of regular functions on the space $\CC^{(r+1)\times (n+1)}$ of matrices. The Chow form $\CF_V$ of the variety $V$ is invariant under the action of $\SL_{r+1}$ on $\CC^{(r+1)\times (n+1)}$ by left multiplication. The first fundamental theorem of invariant theory states that (see, for example, \cite[Theorem~3.2.1]{sturmfelsalgorithms}) every $\SL_{r+1}$ invariant polynomial can be written as a unique bracket polynomial \[ %
F = B([0,1,\dots,r],\dots,[n-r+1,n-r+2,\dots,n+1]) %
\] in the bracket $[i_0,\dots,i_r]$ polynomials. The computation of this representation of $\CF_V$ can be done by the means of Rota's straightening algorithm or the subduction algorithm. We refer to \cite[\S~3]{sturmfelsalgorithms} and, in particular, \cite[Algorithm~3.2.8]{sturmfelsalgorithms} for more information. 

\subsection{The Hurwitz polynomial}
\label{sec:Hurwitz}

Closely related to the Chow form of a projective variety $V \subseteq \PP^n$ is its  Hurwitz form.   
This is a discriminant that characterizes the linear subspaces of dimension $n - r$ that intersect $V$ non-transversally. 
When $\deg(V)\geq 2$, these linear spaces form a hypersurface in the corresponding Grassmannian.
The polynomial of this hypersurface is named as the Hurwitz form of $V$ by Sturmfels \cite{sturmfels_hurwitz_2017}. 
\begin{remark}
  The assumption $\deg(V)\geq 2$ is necessary to have a hypersurface.
  Intuitively, if  $V$ is a linear subspace, then the condition that another linear space intersects $V$ in less than $\deg(V)$ many points actually implies that the
  intersection is empty.
  We should consider this case as a degenerate case of the
  general situation.
\end{remark}

The computation of the Hurwitz form goes along the same lines as the computation
of the Chow form. We assume that $V$ is a complete intersection. We introduce
$r$ linear forms $U_i$ and one linear form $M$, see Alg.~\ref{alg:Hurwitz-homo}.
As $V$ is a complete intersection, if the linear forms are generic, then the resulting system
does not have any solution and hence its resultant is not zero. The resultant of the system when we eliminate the variables $\x$, using the
Poisson formula \cite{cox}, corresponds to the evaluation of $M$ over all the
roots of the system $\{ f_1 = \cdots = f_{n-r} = U_1 = \cdots = U_{r} = 0 \} $.

The (square-free part of the) discriminant of this multivariate polynomial, by
considering it as a polynomial in $\boldsymbol{m}$, the generic coefficients of the linear form $M$, with coefficients in $\ZZ[\u]$ is the Hurwitz polynomial.

\begin{algorithm}[h]      
	\caption{\xspace{ \textsc{HurwitzPoly} }} 
	\label{alg:Hurwitz-homo}
	
	\begin{flushleft}
	\textbf{Input:} $f_1, \dots, f_{n-r} \in \ZZ[\x]$ (complete intersection) \\
	\textbf{Output:} The Hurwitz polynomial of $\VV(f_1, \dots, f_{n-r})$.
	\end{flushleft}
	\begin{enumerate}
         \item Let  
         $U_i = \sum_{j=0}^{n} u_{ij}x_j, \text{ for } i \in [r] $ \;

        \item Let $M = m_0 x_0 + \cdots + m_n x_n$ \;
        \item Eliminate the variables $\x_i$
          \[
          	R_1 = \mathtt{Elim}(\{f_1, \dots, f_{n-r}, U_1, \dots, U_{r}, M\}, \x) \in \ZZ[u_{i,j}][\boldsymbol{m}]
          \]  
        \item Consider the discriminant of $R_1$
          \[
       		R_2 = \mathtt{Elim}(\{ \partial{R_1}/\partial m_0, \dots, 
       		  \partial R_1/\partial m_n \}, \boldsymbol{m}) \in \ZZ[u_{i,j}]
       	  \]
       	\item The Hurwitz polynomial is the square-free part of $R_2$.
    \end{enumerate}
\end{algorithm}

\begin{proposition} \label{lem:Hurwitz-complexity} Consider
  $I = \langle f_{1}, \dots, f_{n-r} \rangle \subseteq \ZZ[x_{0}, \dots, x_{n}]$
  where each $f_{i}$ is homogeneous of degree $d$ and has bitsize $\tau$; also
  the corresponding projective variety, $V$, has pure dimension $r$. Let
  $\u_i := (u_{i,0}, \dots, u_{i, n})$, for $i \in [r]$, be $(n+1)r$ new
  variables.
  The algorithm {\normalfont \textsc{HurwitzPoly}} (Alg.~\ref{alg:Hurwitz-homo})
  correctly computes the Hurwitz polynomial, in the variables $\u_{i}$, of $V$
  in $\sOB((r d)^{(n^2 r^{2})} \tau)$ bit operations.
\end{proposition}

\begin{proof}[Proof sketch]

  The dimension of the variety, $V$, defined by the polynomials $f_{i}$
  is $r$. Thus, by introducing $r$ linear forms, $U_{i}$, if they are sufficiently generic,
  then the (augmented) system becomes zero dimensional
  and the number of solutions is the degree of $V$.
  Furthermore, the introduction of one more linear form, $M$, results
  a system of $n+1$ polynomials in $n+1$ variables; we concetrate on the $\x$ variables.
  The polynomial $M$ plays the role of the $u$-resultant (also appear with the term separating linear form).
  If we eliminate the variables $\x$ from this system, then we obtain a polynomial
  in coefficients of $M$, $R_{1}$, that factors to linear forms.
  The coefficients of the linear forms in this factorization correspond to the solutions of the zero dimensional system. To force (some) of these solutions to have multiplicities
  we compute the discriminant $R_2$ of $R_{1}$.  If this is zero, then there are roots with multiplicities.   The square-free part of $R_{2}$, that is a polynomial in $\u$ (and in the coefficients of the polynomials $f_{i}$) is the Hurwitz form: The reason for this is the prime factorization theorem of \cite{gkz} and the fact that the support set of the polynomials we are working with is a simplex. The computation of $R_1$ is similar to the computation 	of the Chow form of $V$ (Lemma~\ref{thm:CV-homo}).  Then, the computation of $R_2$ results in computing the resultant 	of a square system with polynomials having coefficients
	polynomials in  $\ZZ[u]$ in $r(n+1)$ variables, of degree $\OO(d^r)$
	and bitsize $\OO(d^r \tau)$.
\end{proof}

\section{Multigraded Chow forms}
\label{sec:Chow-mhomo}
In \cite{multigraded}, Osserman and Trager gave a generalization of Chow forms to multiprojetive varieties, i.e., varieties in the \textit{multiprojective space} $\PP^{\bm n}\coloneqq \prod_{i=1}^l \PP^{n_i}$, given as the zero locus of \textit{multihomogeneous polynomials}. The construction of a Chow form in the multiprojective space is similar to the projective case in the sense that the multiprojective Chow form is simply defined as the defining polynomial of the set of linear subspaces of $\PP^{\bm n}$ that intersects the variety. On the other hand, this intersection is dependent on the intersection theory of the variety, i.e., its class in the Chow ring of $\PP^{\bm n}$, which leads to degenerate and non-degenerate cases.

In this section, we will introduce the multigraded associated varieties and provide algorithms to compute them. Moreover, we extend the results of \cite{multigraded} to the multigraded versions of Hurwitz forms.
 
Throughout $\PP^{\bm n}$ denotes the multiprojective space $\PP^{\bm n} =\PP^{n_1}\times\PP^{n_2}\times\dots\times\PP^{n_l}$. For $i=1,\dots,l$, $\x_i=(\x_{i0},\x_{i1},\dots,\x_{in_i})$ denotes the coordinates of $\PP^{n_i}$. We assume that \begin{equation}
\label{eq:defn}
   V = \VV(f_1,f_2,\dots,f_k) \subset \PP^{\bm n}=\PP^{n_1}\times\PP^{n_2}\times\dots\times\PP^{n_l} 
\end{equation} is an $r$-dimensional \textit{multiprojective variety} where each $f_i$ is a \textit{multihomogeneous} polynomial of \textit{multidegree}\[
\bm{d}_i\coloneqq \mdeg(f_i) = (d_{i1},d_{i2},\dots,d_{il}).
\] 

For a vector $\bm{\alpha}\in\NN^l$, $|\bm{\alpha}|\coloneqq\sum_{i=1}^l \alpha_i$ denotes
the $1$-norm of $\bm{\alpha}$. In particular, $|\bm{n}|$ is the dimension of
$\PP^{\bm n}$ and $|\bm{d}_i|$ is the total degree of $f_i$. For
$\bm{\alpha},\bm{\beta}\in\NN^l$, we write $\bm{\alpha} \leq \bm{\beta}$ if $\bm{\beta}$ \textit{dominates} $\bm{\alpha}$, i.e.,
$\forall i\in [l],\, \alpha_i\leq\beta_i$ holds.

\subsection{The multidegree and the support of a multiprojective variety}
A linear subspace of $\PP^{\bm n}=\PP^{n_1}\times\PP^{n_2}\times\dots\times\PP^{n_l} $ is defined to be a product of linear subspaces:
\[
L = L_1 \times L_2\times \dots\times L_l.
\] We say the \textit{format} of $L$ is $\bm{\alpha} =(\alpha_1,\alpha_2,\dots,\alpha_l)$ if $\dim L_i=\alpha_i$. Note that as an abstract variety, $L$ has dimension $|\bm{\alpha}|=\sum_i \alpha_i$.

Contrary to the projective case, for a multiprojective variety, the number of intersection points of a linear subspace of complementary dimension may vary with the format of the subspace. Recall that a projective variety of degree $d$ has $d$ intersection points with a linear subspace of complementary dimension. The simplest counter-examples in the multiprojective space occur when one considers the multiprojective varieties that are products of projective varieties. For example, if $V=\PP^{1}\times \{p\}\subset\PP^1\times\PP^1$ then the intersection with linear subspaces of format $(1,0)$ is generically empty whereas for a linear subspace of format $(0,1)$, the intersection is a singleton.

This observation leads to the following definition which aims to capture the intersection theoretic properties of the multiprojective variety $V$. 

\begin{definition}
  Let $V\subset\PP^{\bm n}$ be a pure dimensional multiprojective variety of
  dimension $r$. The \textit{support} $\supp(V)$ of $V$ is the set
  of all formats $\balpha\in\NN^l$ such that $|\balpha|=\codim V$ and the
  intersection \[ V\cap (L_1\times L_2\times\dots\times L_l)
	\] of $V$ with a generic linear subspace $L=(L_1,L_2,\dots,L_l)$ of format $\balpha$ is non-empty.
	
    The \textit{multidegree}\footnote{This concept is also called the \textit{dimension} or the \textit{multidimension} of a multiprojective variety in the literature. We reserve the term dimension for the dimension of $V$ as a projective variety, e.g., the dimension of its Segre embedding.} $\mdeg(V)$ of $V$ is the set of all tuples
    $(m_{\balpha}, \balpha)$ where $\balpha\in\supp(V)$ and $m_{\balpha}$ is the
    number of intersection points of $V$ with a generic linear subspace of
    $\PP^{\bm n}$ of format $\balpha$.
\end{definition}
We note that the number of intersection points, $m_{\balpha}$, is finite for dimension reasons.
\begin{example}
\label{ex:product}
For $i \in [l]$, let $V_i\subset\PP^{n_i}$ be projective varieties of dimension $r_i$ and set \[
V\coloneqq V_1\times V_2\times\dots\times V_l\subset\PP^{\bm{n}}.
\] For a linear subspace $L=L_1\times L_2\times\dots\times L_l$, the equality  \[
V\cap L = (V_1\cap L_1)\times (V_2\cap L_2)\times\dots\times (V_l\cap L_l)
\] clearly holds. Hence, setting $\bm{\beta}\coloneqq (n_1-r_1,n_2-r_2,\dots,n_l-r_l)$, one can observe that $V$ intersects with a generic subspace of format $\bm{\beta}$ at $\prod_{i=1}^l \deg V_i$ many points. For any other format $\bm{\gamma}$ with $|\bm{\gamma}|=\codim V=|\bm{n}|-\dim V$, there exists $i$ such that $\gamma_i+r_i<n_i$ and, thus; the intersection $V_i\cap L_i$ is generically empty. The equalities \[
\supp(V)=\{\bbeta\},\quad\mdeg(V) = \{(\prod_{i=1}^l \deg(V_i), \bbeta)\}
\] follow.
\end{example}

\begin{remark}
  The multidegree of $V$ can be also seen as the \textit{class} of $V$ in the
  \textit{Chow ring} of $\PP^{\bm n}$. Informally, the Chow ring of
  $\PP^{\bm n}$is the set of all formal linear combinations of the subvarieties
  of $\PP^{\bm n}$, modulo the relations given by \textit{rational equivalences}
  (see, for example, \cite[Definition~1.3]{eisenbudharris}). In the case of the
  multiprojective space $\PP^{\bm n}$, the Chow ring is generated by the
  \textit{cycles} of the form
  \[
    [L]=[L_1\times L_2\times\dots\times L_l]
  \]
  where each $L_i$ is a linear subspace of $\PP^{n_i}$. Two cycles $[L]$ and $[L']$ are equal if $L$ and $L'$ have the same format.

For a multiprojective variety $V\subset\PP^{\bm n},$ the statement \[
\mdeg(V) = \{(d_1,\bbeta_1),(d_2,\bbeta_2),\dots,(d_k,\bbeta_k)\},
\] is equivalent to the statement that \[
[V] = \sum_{i=1}^k d_i [L_i]
\] in the Chow ring of $\PP^{\bm n}$, where $L_i=[L_{i1}\times L_{i2}\times\dots\times L_{il}]$ has format $\bbeta_i$.
\end{remark}

For an index set $\varnothing\neq I\subset [l]$, let \[
\pi_I : \prod_{i=1}^l \PP^{n_i}\rightarrow \prod_{i\in I}\PP^{n_i}
\] denote the projection of $\PP^{\bm n}$ onto $\prod_{i\in I}\PP^{n_i}$. The main result of \cite{clz} is that the support of $V$ can be easily computed if one is given $\dim \pi_I(V)$ for each index set $\varnothing\neq I\subset [l]$.
\begin{theorem}[\cite{clz}]
\label{thm:supp}
	Assume $V$ is an irreducible variety in $\PP^{\bm n}$. Let $\bbeta\in\NN^l$ with $\bbeta\leq \bm{n}$ and $|\bbeta| = \codim V$. Then, $\bbeta\in\supp(V)$ if and only if for all $\varnothing\neq I\subset [l]$ we have
	\[
	\sum_{i\in I} (n_i-\beta_i) \leq \dim \pi_I(V).
	\]
\end{theorem}
The above result can also be generalized to non-irreducible varieties as for a pure-dimensional multiprojective variety $V$ with decomposition $V=V_1\cup\dots\cup V_k$, one has $\supp(V) = \cup_{i=1}^k \supp(V_i)$. See \cite[Corollary~3.13]{clz} for the exact statement. 
\begin{remark}
\label{rem:submodular}
Assume $V$ is irreducible. Define the function $\delta$ from the power set $2^{[l]}$ of $[l]$ to $\mathbb{N}$ via $\delta(I)=\dim\pi_I(V)$ for $\varnothing\neq I\subset [l]$ and $\delta(\varnothing)=0$. Then $\delta$ is a $\textbf{submodular function}$, meaning $\delta$ satisfies the following properties:
\begin{enumerate}
    \item $\delta(\varnothing)=0,$
    \item for $I\subset J$, $\delta(I)\leq \delta(J)$, and,
    \item for $I,J\subset [l]$, $\delta(I)+\delta(J)\geq \delta(I\cap J)+\delta(I\cup J)$.
\end{enumerate}
The proof of this fact can be found in \cite{clz}. We also refer to \cite{hauenstein2021numerical} for more on the combinatorial structure of $\supp(V)$. We note that together with Theorem~\ref{thm:supp}, the submodularity of $\dim\pi_I(V)$ implies that $\supp(V)\subset\NN^l$ is a polymatroid. We will discuss more on this in Section~\ref{sec:fun}.
\end{remark}

%

\subsection{Computing the support of a multiprojective variety}

In this section, we provide algorithms to compute $\supp(V)$ by the means of Theorem~\ref{thm:supp}. The idea is to compute $\dim\pi_I(V)$ for every $I\subset [l]$ and iterate through each possible format $\alpha\in\NN^l$ of dimension $\codim(V)$ and test membership to $\alpha\in\supp(V)$. Here, we emphasize the crucial observation that we do not have access to the defining equations of $\pi_I(V),$ since this requires elimination of variables and significantly increases the complexity. Instead, we will show that a small modification of the original dimension algorithms (\cite{koiran_randomized_1997,chistov}) can be used to compute the dimension of any linear projection $\pi(V)$ of a variety, $V$.

For simplification, we will assume that $V\subset\CC^{n}$ is an affine variety.
To compute the dimension, there is no harm in working with affine varieties
compared to projective/multiprojective ones since we can always consider the
(multi)affine cone
\[
  V_{\Aff}\subset\prod_{i=1}^l \CC^{n_i+1}=\CC^{|n|+l}
\]
over $V\subset\prod_{i=1}^l \PP^{n_i}$, defined as the zero set of the same set of polynomials, $f_1,f_2,\dots,f_k$. Then the dimension of the (multi)affine cone
and the multiprojective variety is related by the formulas
$\dim V_{\Aff}=\dim V+l$ and $\dim\pi_I(V_{\Aff})=\dim\pi_I(V)+|I|$, and, in
particular, we can compute $\dim\pi_I(V)$ from $\dim\pi_I(V_{\Aff})$.

The dimension algorithms in \cite{koiran_randomized_1997,chistov} rely on the observation that a variety $Z\subset\CC^n$ has dimension at least $s$ if and only if a generic affine subspace $L\subset\CC^n$ of dimension $n-s$ intersect $Z$. Now we take $Z=\pi_I(V)$. Assume $m\leq n$ and $\pi:\CC^n\rightarrow\CC^m$ denotes the orthogonal projection onto the first $m$ coordinates. Then, for an affine subspace $L\subset\CC^m$ we have \[
\pi(V)\cap L\neq\varnothing\;\iff\;V\cap\pi^{-1}(L)\neq\varnothing.
\] In particular, $\dim\pi(V)\geq s$ if and only if for a generic affine subspace $L\subset\CC^m$ of dimension $m-s$ we have $V\cap\pi^{-1}(L)\neq\varnothing$. Note that $\pi^{-1}(L)$ is an affine subspace of dimension $n-s$ and can be given as the zero locus of $s$ linear polynomials. The complexity of computing $\dim\pi_I(V)$ is hence equivalent to the complexity of constructing a generic linear subspace $L\subset\CC^m$ (see \cite[Lemma~5.5, Theorem~5.6]{koiran_randomized_1997}) and checking if $f_1,f_2,\dots,f_k$ have a common zero in $L$. Following \cite{koiran_randomized_1997}, the complexity of these tasks is bounded by $k^{\OO(1)}d^{\OO(n)}\OO(\tau)$.
\begin{theorem}
       Assume $V\subset\CC^n$ is given as the zero set of polynomials $f_1,f_2,\dots,f_k\in\CC[{\bm x}]$ of degree $\leq d$ with integer coefficients of bitsize $\leq\tau$. Let $m\leq n$. Then, the dimension of the image $\pi(V)$ of $V$ under the orthogonal projection $\pi:\CC^n\rightarrow\CC^m$ onto the first $m\leq n$ coordinates can be computed in \[
       k^{\OO(1)}\, d^{\OO(n)}\, \OO(\tau)
       \]
       bit operations.
\end{theorem}

\begin{theorem}
Assume $V\subset\PP^{\bm n}=\prod_{i=1}^l \PP^{n_i}$ is an irreducible multiprojective variety of dimension $r$, given as the zero set of multihomogeneous polynomials $f_1,f_2,\dots,f_k$, of multidegrees $\bm{d}_1,\bm{d}_2,\dots,\bm{d}_k$ and with integral coefficients of bitsize bounded by $\tau$.

Then, we can compute $\supp(V)$ in time 
\[
k^{\OO(1)}\, D^{\OO(|\bm{n}|)}\,2^l\, \OO(\tau) + 2^l \OO(\binom{|\bm{n}|-r}{l})
\] 
bit operations where $D=\max_{i} \{|\bm{d}_i|\}$ is the maximum total degree of $f_i$. 
\end{theorem}
\begin{proof}
Using the previous theorem, for each $I\subset [l]$ we can compute $\dim\pi_I(V)$ in \[
k^{\OO(1)}\, D^{\OO(|\bm{n}|)}\, \OO(\tau).
\] Iterating through the power set $2^{[l]}$, the family $(\dim\pi_I(V)\mid I\subset [l])$ can be computed in the claimed complexity. Using Theorem~\ref{thm:supp}, we can now compute $\supp(V)$ by iterating through each possible format $\alpha$ with $|\alpha|=|\bm{n}|-r$ and decide whether for every $I\subset [l]$ $\sum_{i\in I} n_i-\alpha_i\leq\dim\pi_I(V)$ holds. The number of possible formats is bounded by $\binom{|\bm{n}|-r}{l}$ and the number of constraints to be checked is bounded by $2^l$.
\end{proof}

\subsection{Associated varieties of multiprojective varieties} In this section, we introduce the generalization of associated hypersurfaces to multiprojective varieties. The definitions and the results of this section follow \cite{multigraded}.

\begin{definition}
Let $\alpha\in\NN^l$ be a format such that $\alpha\leq\bm{n}$, i.e., $\forall i\in [l],\,\alpha_i\leq n_i$, and $|\alpha|= \codim V-1$. The \textit{associated variety of $V$ of format $\alpha$} is defined to be the multiprojective variety \[
\CZ_{V,\alpha} =\Big\{\,(L_1,L_2,\dots,L_l)\in\prod_{i=1}^l \Gr(\alpha_i,n_i)\;\mid\; V\cap (L_1\times L_2\times\dots\times L_l)\neq\varnothing\,\Big\}.
\]	That is, $\CZ_{V,\alpha}$ is the set of all linear subspaces of $\PP^{\bm{n}}$ of format $\alpha$ that intersect $V$. 
\end{definition}
As the term \textit{associated variety} suggests, $\CZ_{V,\alpha}$ is not always a hypersurface. 
\begin{example}
\label{ex:chow}
Assume $n_1,n_2\geq 3$, let $V_1\subset\PP^{n_1},V_2\subset\PP^{n_2}$ be arbitrary varieties of codimension $2$ and consider $\tilde{V} \coloneqq V_1\times V_2\subset\PP^{n_1}\times\PP^{n_2}$. Since $\codim\tilde{V}=4$, there are four possible formats for the associated varieties, namely $\alpha=(3,0),(2,1),(1,2),(0,3)$. By the symmetry of $V_1,V_2$, we will only consider $(3,0)$ and $(2,1)$. Note that \[
\begin{split}
\CZ_{\tilde{V},(2,1)} &= \{(L_1,L_2)\in\Gr(2,n_1)\times\Gr(1,n_2) \mid L_1\times L_2\cap V\neq\varnothing\}\\
&=\{(L_1,L_2)\in\Gr(2,n_1)\times\Gr(1,n_2)\mid L_2\cap V_2\neq\varnothing\}\\
&=\Gr(2,n_1)\times \CZ_{V_2}
\end{split}
\] is indeed a hypersurface. However, \[
\begin{split}
\CZ_{\tilde{V},(3,0)} &=\{(L_1,p)\in\Gr(3,n_1)\times \PP^{n_2}\mid (L_1\times\{p\})\cap V\neq\varnothing\}\\
&=\{(L_1,p)\in\Gr(3,n_1)\times \PP^{n_2}\mid p\in V_2\}\\
&=\Gr(3,n_1)\times V_2
\end{split}
\] is a codimension $2$ variety in $\Gr(3,n_1)\times\PP^{n_2}.$ 
\end{example}
The formats in $\supp(V)$ and the formats for which the associated variety is a hypersurface are closely related. If we consider the previous example, the support $\supp(V)$ of $V$ is $\{(2,2)\}$ by Example~\ref{ex:product}, and the formats $\alpha$ for which $\CZ_{V,\alpha}$ is a hypersurface are $(2,1)$ and $(1,2)$. If we mark the formats in the support of $V$ and the formats where $\CZ_{V,\alpha}$ is a hypersurface, we arrive at the following diagram in the partially ordered set of the formats:

\begin{figure*}[h]
\centering
\begin{tabular}{cc}
\begin{tikzpicture}[scale = 0.7]

\draw [thin, gray] (0,0) grid (4,4);

\draw [<->,thick] (0,4) node (yaxis) [above] {$\alpha_2$}
|- (4,0) node (xaxis) [right] {$\alpha_1$};

\draw[-,dashed] (-0.4,4.4) -- (4.4,-0.4);

\draw[-,dashed] (-0.4,2) -- (4.4,2);

\draw[-,dashed] (2,-0.4) -- (2,4.4);

\coordinate (s) at (2,2);
\fill[red] (s) circle (3pt);
\node at (2.7, 2.5) {$(2,2)$} ;

\end{tikzpicture}    
    
   \qquad  &  \qquad
     
     \begin{tikzpicture}[scale = 0.7]

\draw [thin, gray] (0,0) grid (4,4);

\draw [<->,thick] (0,4) node (yaxis) [above] {$\alpha_2$}
|- (4,0) node (xaxis) [right] {$\alpha_1$};

\draw[-,dashed] (-0.4,3.4) -- (3.4,-0.4);

\draw[-,dashed] (1,-0.4) -- (1,4.4); 

\draw[-,dashed] (-0.4,1) -- (4.4,1);

\coordinate (s) at (2,2);
\fill[red] (s) circle (3pt);

\coordinate (c1) at (2,1);
\fill[blue] (c1) circle (3pt);
\node at (1.7, 2.5) {$(1,2)$} ;

\coordinate (c2) at (1,2);
\fill[blue] (c2) circle (3pt);
\node at (2.7, 1.5) {$(2,1)$} ;

\end{tikzpicture}
\end{tabular}
    \caption{Example~\ref{ex:chow}. On the left, we have $\supp(V)$, cut out by $\alpha_1+\alpha_2=4,\alpha_1\geq 2,\alpha_2\geq 2$, as described in Theorem~\ref{thm:supp}. On the right, we have the set of formats such that the associated variety is a hypersurface which is cut out by $\alpha_1+\alpha_2=3,\alpha_1\geq 1,\alpha_2\geq 1$.}
    \label{fig:chow}
\end{figure*}
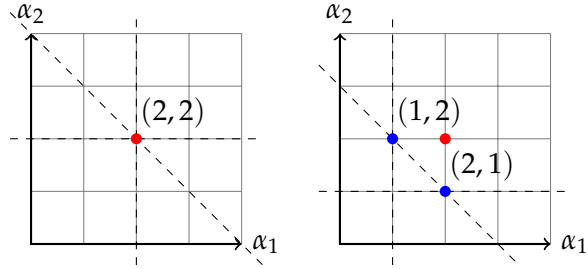

This example is no coincidence and the next proposition clarifies the relation between $\supp(V)$ and the formats of associated varieties.

\begin{proposition}
\label{prop:formats}
	Assume $\alpha\leq \bm{n}$ and $|\alpha|=\codim V-1$. Then the following are equivalent.
	\begin{enumerate}
	\item $\CZ_{V,\alpha}$ is a hypersurface.
	\item There exists $\beta\in\supp(V)$ such that $\alpha\leq\beta$.
	\item For all $\varnothing\neq I\subset [l],$ we have $\dim\pi_I(V)\geq \sum_{i\in I}(n_i-\alpha_i)-1$.	
	\end{enumerate}

\end{proposition}
\begin{proof}
See Proposition~3.1 and Corollary~5.11 of \cite{multigraded}.
\end{proof}

\subsection{Computing the Chow form of a multiprojective variety}
In this section, we provide algorithms to compute associated varieties of multiprojective varieties. For the rest of the section, $V\subset\PP^{\bm{n}} = \PP^{n_1}\times\PP^{n_2}\times\dots\times\PP^{n_l} $ is a pure dimensional multiprojective variety of dimension $r$.

\subsubsection{The complete intersection case}
\begin{algorithm}[h]      
	\caption{\xspace{ \textsc{MultiChowForm\_CI} }} 
	\label{alg:Multi-Chow-homo-CI}
	
	\begin{flushleft}
	\textbf{Input:} $f_1, \dots, f_{|\bm{n}|-r} \in \ZZ[\x_1,\x_2,\dots,\x_l], \alpha\in\NN^l$ \\
	\textbf{Precondition: } $V=\VV(f_1,\dots,f_{|\bm{n}|-r})$ is pure $r$-dimensional. $\CZ_{V,\alpha}$ is a hypersurface.\\
	\textbf{Output:} The Chow form of $V$ corresponding to format $\alpha$.
	\end{flushleft}
	\begin{enumerate}
    	\item Consider linear forms,  \[
    	  U^{i}_j \coloneqq \sum_{k=0}^{n_i} u^i_{jk}x_{ij}, \text{ for } 0\leq j \leq n_i-\alpha_i-1  
         \] for $i=1,2,\dots,l$.
        \item Eliminate the variables $\x_1,\dots,\x_l$.
          \[
          	R = \mathtt{Elim}(\{f_1, \dots, f_{|\bm{n}|-r}, U^i_j\}, \{\x_1,\dots,\x_l \}) \in \ZZ[u^i_{jk}]
          \]
	    \item $R_r = \textsc{SquareFreePart}(R)$.
        \item \textsc{return} $R_r$.
    \end{enumerate}
\end{algorithm}

As in the case of projective varieties, we first assume that $V$ is a complete intersection, i.e., \[
V = Z(f_1,f_2,\dots,f_{|\bm{n}|-r})
\] is the zero locus of $k=|\bm{n}|-r$ many multihomogeneous polynomials. To simplify notation for the next lemma, we will denote $n =|\bm{n}|=\sum_{i=1}^l n_i$ and assume that $\forall i\in [k], j\in [l], \deg(f_i;\x_j)\leq d$. 

\begin{lemma}
	\label{lem:CV-mhomo-CI}
	Let $V$ be a $r$-dimensional complete intersection, i.e., $V$ is the zero locus 
	$\VV(f_1,f_2,\dots,f_{n-r})$  of $k\coloneqq n-r$ many multihomogeneous polynomials
   and assume that $\deg(f_i;\x_j)\leq d$ for $i\in [k], j\in [l]$ and the bitsizes of $f_i$ are bounded by $\tau$.
   Set $B_r := (dl)^{n-r}\sum_{i=1}^l \binom{n-r}{\alpha_1, \dots,\alpha_i+1,\dots,\alpha_l} $, where the summands in the second factor are the multinomial coefficients.
	If $\alpha$ is a format that satisfies the equivalent conditions of Proposition~\ref{prop:formats}, then the Chow form of $V$ corresponding to $\alpha$ is a multihomogeneous polynomial
	in $A = \sum_{i=1}^{l}(n_i-\alpha_i)(n_i + 1)$ new variables of total degree at most $B_r$
	and bitsize $\sO(n B_r \tau)$.
	{\normalfont \textsc{MultiChowForm\_CI}} (Alg.~\ref{alg:Multi-Chow-homo-CI}) computes $\CF_{V, \alpha}$ in 
	\[
	\sOB(n^{\omega+1} 2^{(\omega+1)n} \, B_r^{(\omega+1)r^2 + 2 A + \omega+1} 
	\, (\tau + n^3))
	\]
	bit operations (Las Vegas) where $\omega$ is the exponent of matrix multiplication,	
\end{lemma}
\begin{proof}
The proof is similar to the proof of Proposition~\ref{prop:CV-homo-CI}. To exploit the multihomogeneity, we use sparse/multiprojective resultant computations
and  the multihomogeneous B\'ezout bound. We have $n-r$ multihomogeneous polynomials,
each having (total) degree at most $dl$.  To compute the $\CF_{V,\alpha}$ we add  $(n_i-\alpha_i)$ linear forms in $\x_i$, $U_{j}^i$;
their coefficients are the variables $u^i_{jk}$, for $i \in [l]$. The sparse resultant is an irreducible polynomial in the coefficients of $f_i$ and $u^{i}_{jk}$, which vanishes if and only if $V$ and the linear subspace described by $U^{i}_j$ intersects. 

The sparse resultant is homogeneous in each set of variables $\bm{u}^i_j$. Its degree with respect to $\bm{u}^i_j$ equals to the generic number of solutions of the remaining system when $\bm{u}^i_{j}$ is omitted, hence bounded by the multihomogeneous B\'{e}zout bound $
(dl)^{n-r}\binom{n-r}{\alpha_1,\dots,\alpha_i+1,\dots,\alpha_l}$ where the second factor is the \textit{multinomial coefficient}:\[
\binom{N}{a_1,\dots,a_l}\coloneqq \frac{N!}{a_1 ! \, a_2! \dots \, a_l!}
\] Thus, the resultant has total degree at most $B_r$. The coefficients are integers of bitsize $\sO(n B_r \tau)$. The number of monomials is bounded by $\sO(B_r^{r^2})$. Following \cite{DA-spres-02}, we compute the sparse resultant as a ratio of two determinants using the sparse resultant matrix. The sparse resultant matrix has dimension $M \times M$,
where $M = \OO(n e^n B_r)$ and each entry of $M$ is a coefficient of one of the input polynomials $f_i$ or the linear forms $U^{i}_{j}$.

It suffices to specialize $u^i_{jk}$  to numbers of bitsize $\OO(n^3 \lg(d))$. So the specialized matrix contains numbers of bitsize $\sO(\tau + n^3)$. We compute each determinant in $\sOB(M^{\omega+1}(\tau+n^3))$. We need to perform this computations $\sO(B_r^{r^2})$ many times and then we recover the resultant using interpolation. The cost of all the evaluations is $\sOB(n^{\omega+1} 2^{(\omega+1)n}  B_r^{r^2 + \omega+1} (\tau + n^3))$. The cost of interpolation is $\sOB( B_r^{\omega r^2 + 1} (\tau + n^3))$. Finally, the cost of computing the square-free part is 
$\sOB( (n  B_r)^{2A + 1} \tau)$.

The algorithm is of Las Vegas type becauce of the construction of the resultant matrix (see the remark that follows).
\end{proof}

\begin{remark}
In the previous complexity estimate, we should also take into account the cost for constructing 
the sparse resultant matrix.
Following \cite[Thm.~11.6]{CE-subdiv-jacm} there is a Las Vegas algorithm
for this computation with cost $\sOB( M)$, 
where $M$ is the size of the matrix and the number of lattice points
in the Minkowski sum of the Newton polytopes of the input polynomials.
In our case, as the polynomials are multihomogeneous, 
the corresponding Newton polytopes are product of simplices. 
Therefore, we can also afford to construct the resultant matrix using the lower
hull of an appropriate (sufficiently generic) lifting of the lattice points of the Minkowski sum of
the Newton polytopes; this costs $\sO(M^{\lfloor n/2 \rfloor})$ \cite{ps-cg-12}.
Neither complexity bound dominates the overall complexity; this is so
because the resultant matrix contains polynomials in many variables, 
that is the  $\bm{u}^i_j$'s. 

The Las Vegas characterization is due to the sufficiently generic lifting but also due to the random perturbation needed in order to assign the lattice points to the appropriate polynomials. We refer to \cite{CE-subdiv-jacm,DA-spres-02} for further details.

To avoid the case that the denominator is zero in the resultant computations, we can apply the technique of the generalized characteristic polynomial \cite{canny_generalised_1990}, similarly to the projective case. Now, we can apply a symbolic perturbation to all the terms of all the polynomials \cite{rojas1999solving}
or only to the terms that appear in the diagonal of the resultant matrix \cite{srur-17}. 
In both cases, we introduce one additional variable that does not affect the single exponential behavior of the complexity bound.

\end{remark}

\subsubsection{The general case}
\begin{algorithm}[h]      
	\caption{\xspace{ \textsc{MultiChowForm} }} 
	\label{alg:Multi-Chow-homo}
	
	\begin{flushleft}
	\textbf{Input:} $f_1, \dots, f_{m} \in \ZZ[\x_1,\x_2,\dots,\x_l], r\in\NN, \alpha\in\NN^l$  \\
	\textbf{Precondition:} $V=\VV(f_1,\dots,f_m)$ is pure $r$-dimensional and $\CZ_{V,\alpha}$ is a hypersurface.\\
	\textbf{Output:} The Chow form of $V$.
	\end{flushleft}
	\begin{enumerate}
	\item $\Lambda^1,\dots,\Lambda^N \coloneqq \textsc{MultiGenericLC}(f_1,f_2,\dots,f_m)$.
    	\item \lFor{$r \in [N]$}{$F_i = \textsc{ChowForm\_CI}(\Lambda_{\ff}^i)$}
    	      
        \item \textsc{return} $\gcd( F_1, \dots, F_N )$ 
    \end{enumerate}
\end{algorithm}

Now we remove the assumption that $V$ is a complete intersection and assume \[
V = Z(f_1,f_2,\dots,f_{m})\subset\PP^{\bf{n}},
\] where $m\geq |\bm{n}|-r$. Consider the multidegrees 
$\bm{d}^i = \mdeg(f_i)$ and set \[
\bm{d} = (\max_i d^{i}_1, \max_{i} d^{i}_2,\dots,\max_i d^{i}_m).
\] Note that for all $i=1,2,\dots,m$ we have $\bm{d}^i \leq \bm{d}$ by construction. For each $i=1,2,\dots,m$ with $\bm{d}^i<\bm{d}$, we replace the polynomial $f_i$ with the collection $\bm{x}_1^{\alpha_1} \bm{x}_2^{\alpha_2} \cdots \bm{x}_l^{\alpha_l} f_i $
where $\bm{x}_j$ denotes the $j$-th block of variables of the multiprojective space $\PP^{\bm{n}}$ and $\alpha_j$ runs over the all possible monomials with $|\alpha_j| = \bm{d}_j - \bm{d}^i_j$. The new collection $\bm{\tilde{f}}$ has the property that each polynomial in it has the same multidegree, $\bm{d}$. Hence, without loss of generality, we will assume throughout the rest of the section that $V=Z(f_1,f_2,\dots,f_m)$ where each $f_i$ has the same multidegree, $\mdeg(f)=\bm{d}$. 

As in the projective case, for $\Lambda\in\CC^{k \times m}$ we consider $k$ linear combinations $\Lambda_{\ff}$ of $\ff$, defined as in (\ref{eq:lambdasystem}). By the assumption that each $f_i$ has the same multidegree, each linear combination has multidegree $\bm{d}$, and, thus; has a well-defined zero locus in $\PP^{\bm{n}}$.

For generic $\Lambda\in\CC^{k\times m}$ we have $
Z(\Lambda_{\f}) = V\cup X $ for some pure dimensional variety $X$ of dimension $|\bm{n}|-k$. The proof is essentially the same as the projective case, Proposition~\ref{prop:genericlambda} and Proposition~\ref{prop:number}. The only change in the proof is the bound on the number of irreducible components of a variety, where the B\'{e}zout bound is replaced by the multihomogeneous B\'{e}zout bound. 
\begin{proposition}
\label{prop:multigenericlambda}
Let $N=\lceil \frac{m}{|\bm{n}|-r}\rceil$. For generic choices of matrices $\Lambda^1,\dots,\Lambda^N\in\CC^{(|\bm{n}|-r)\times m}$, each variety $\VV(\Lambda^i_{\ff})$ is a pure dimensional variety of dimension $r$ and $V=\cap_{i=1}^N \VV(\Lambda^i_{\ff})$. More concretely, there is a hypersurface $H\subset\CC^{N(|\bm{n}|-r)\times m}$ of degree $\leq N(|\bm{n}|-r)|\bm{d}|^{|\bm{n}|-r-1}+m$ such that for any $(\Lambda^1,\dots,\Lambda^N)\in \CC^{N(|\bm{n}|-r)\times m}\setminus H$, the condition is satisfied.
\end{proposition}

The proposition allows us to consider the following algorithm to generate $(\Lambda^1,\dots,\Lambda^N)$ satisfying the condition of Proposition~\ref{prop:multigenericlambda}.

\begin{algorithm}[h]      
	\caption{\xspace{ \textsc{MultiGenericLC} }} 
	\label{alg:MultiGenericLambda}
	
	\begin{flushleft}
	\textbf{Input:} $f_1, \dots, f_{m} \in \ZZ[\x_1,\x_2,\dots,\x_l], r\in\NN$  \\
	\textbf{Precondition: } $\VV(f_1,f_2,\dots,f_m)$ is pure $r$-dimensional.\\
 	\textbf{Output:} $\Lambda^1,\Lambda^2,\dots,\Lambda^N$.\\
 	\textbf{Postcondition:} See Proposition~\ref{prop:multigenericlambda}.
	\end{flushleft}
	\begin{enumerate}
    	\item $N\coloneqq\lceil \frac{m}{|\bm{n}|-r}\rceil$.
    	\item $S\coloneqq [N(|\bm{n}|-r)|\bm{d}|^{|\bm{n}|-r}+m+1]\subset\NN$.
    	\item \lFor{$(\Lambda^1,\Lambda^2,\dots,\Lambda^N) \in S^{N(|\bm{n}|-r)m}$}{
    	\\ \textbf{if} $\dim(\VV(\Lambda^i_{\ff}))\leq r$ and $\Xi$ is full-rank \textbf{then}
        \\\textsc{return} $\Lambda^1,\dots,\Lambda^N$}
    \end{enumerate}
\end{algorithm}

\begin{lemma}
Algorithm \ref{alg:MultiGenericLambda} returns a tuple $(\Lambda^1,\Lambda^2,\dots,\Lambda^N)$ satisfying the requirements of Proposition~\ref{prop:multigenericlambda} in $\tau m^{2m^2+\OO(1)}|2\bm{d}|^{m^2n+\OO(|\bm{n}|)}$.
\end{lemma}
\begin{proof}
The proof goes as in Lemma~\ref{lem:glc-alg}. To test dimension of $\VV(\Lambda^i_{\ff})$, we consider the affine cone $C=\VV_{\Aff}(\Lambda^i_{\ff})$ over $\VV(\Lambda^i_{\ff})$. We have $\dim C = \dim \VV(\Lambda^i_{\ff})+l$, so we can compute the dimension of $\VV(\Lambda^i_{\ff})$ from $\dim C$.
\end{proof}

For the simplicity of notation, we will assume for the next theorem that $d = \max_i \bm{d}_i$ and $n=|\bm{n}|$.

\begin{theorem}
	\label{thm:MultiChowForm}
	Consider $I = \langle f_{1}, \dots, f_{m} \rangle \subseteq \ZZ[\x_{1}, \dots, \x_{l}]$,
	where each $f_{i}$ is multihomogeneous of degree $\bm{d}$ and bitsize $\tau$;
	also the corresponding multiprojective variety, $V$, has pure dimension $r$. Also, $B_r = (dl)^{n-r}\sum_{i=1}^l \binom{n-r}{\alpha_1, \dots,\alpha_i+1,\dots,\alpha_l} $.

	The Chow form of $V$ corresponding to a format $\alpha$ 
	is a multihomogeneous polynomial
	in $A = \sum_{i=1}^{l}(n_i-\alpha_i)(n_i + 1)$ new variables of total degree $ \leq B_r$
	and bitsize $\sO(n B_r \tau)$.

	{\normalfont \textsc{MultiChowForm}} (Alg.~\ref{alg:Multi-Chow-homo}) computes $\CF_V$ in 
		\[
		\sOB(m^{2m^2+\kappa} n^{\omega+1} 2^{(\omega+1)n}\, B_r^{m^2n+(\omega+1) r^2+2A+\omega+1}\, (\tau+n^3)),
		\]
	bit operations where $\omega$ is the exponent of matrix multiplication and $\kappa$ is a small constant, depending on the precise complexity of the dimension test in Alg.~\ref{alg:MultiGenericLambda}.
\end{theorem}
\begin{proof}
The cost of generating $\Lambda^1,\dots,\Lambda^N$ is $\tau m^{2m^2+\OO(1)}(2d)^{m^2n+\OO(n)}$ by the previous lemma. 

$\Lambda^i_{\ff}$ have bitsizes bounded by $\OO(\lg m+n\lg d+n\lg l+\tau)=\sO(\tau+n)$. As there are $N=\OO(m)$ Chow forms to compute, the second step costs $\sOB(m n^4 2^{4n}B_r^{(\omega+1)r^2+2A+6}\, (\tau^2+n^6))$.

For the last step, we need to compute the gcd of $N$ Chow forms. As in the proof of Theorem~\ref{thm:CV-homo}, the cost of this step is less than the claimed complexity, therefore we can omit it.
\end{proof}

\subsection{The multiprojective Hurwitz form}
Recall that for a projective variety $V\subset\PP^n$ of dimension $r$, the Hurwitz form is defined to be the defining polynomial of the set of all linear forms $L\in\Gr(n-r,n)$ such that the intersection $L\cap V$ is non-generic, i.e., either $L\cap V$ is infinite, or, $L\cap V$ is finite but $|L\cap V|<\deg V$. As in the case of the Chow forms, one readily generalizes the Hurwitz form to multiprojective varieties.

\begin{definition}
Assume $V\subset\PP^{\bm n}$ is an irreducible multiprojective variety of dimension $r$ and $\alpha\in\NN^l$ is a format with $|\alpha|=\codim V$. We define the higher associated variety $\mathcal{HZ}_{V,\alpha}\subset\PP^{\bm n}$ as the set of all linear subspaces $L$ of format $\alpha$ which intersects $V$ in non-generic way.  That is, for $\alpha\in\supp(V)$, $\mathcal{HZ}_{V,\alpha}$ is the set of all linear subspaces $L$ of format $\alpha$ such that $V\cap L$ is either infinite, or, $|L\cap V|<\mdeg(V,\alpha)$. Similarly, if $\alpha\not\in\supp(V)$, then we define $\mathcal{HZ}_{V,\alpha}$ as the set of all linear subspaces $L$ of format $\alpha$ such that $L\cap V\neq\varnothing$.
\end{definition}
As in the case of the Chow forms, the higher associated variety is not always a hypersurface.
\begin{example}
\label{ex:hurwitz}
Recall Example~\ref{ex:chow}, where \[
\tilde{V} = V_1\times V_2\subset\PP^{n_1}\times\PP^{n_2}
\] for codimension $2$ varieties $V_1,V_2$. Moreover, assume that $\deg V_1,\deg V_2>1$. Note that we have $\supp(V)=\{(2,2)\}$. For $\alpha=(4,0)$, \[
\begin{split}
    \mathcal{HZ}_{\tilde{V},(4,0)} &= \{L_1\times\{p\}\in\Gr(4,n_1)\times\Gr(0,n_2)\mid (L_1\times\{p\})\cap \tilde{V}\neq\varnothing\}\\
    &=\{L_1\times\{p\}\in\Gr(4,n_1)\times\Gr(0,n_2)\mid p\in V_2\}\\
    &=\Gr(4,n_1)\times V_2
\end{split}
\] has codimension $2$. For $\alpha=(3,1)$, on the other hand, \[
\begin{split}
    \mathcal{HZ}_{\tilde{V},(3,1)} &= \{L_1\times L_2\in\Gr(3,n_1)\times\Gr(1,n_2)\mid L_2\cap V_2\neq\varnothing\}\\
    &=\Gr(3,n_1)\times\CZ_{V_2}
\end{split}
\] is a hypersurface. For $\alpha=(2,2)$, \[
\begin{split}
    \mathcal{HZ}_{\tilde{V},(2,2)} &= \{L_1\times L_2\in\Gr(2,n_1)\times\Gr(2,n_2)\mid \sharp(L_1\times L_2\cap \tilde{V})\neq \deg V_1\deg V_2\}\\
    &=\Big(\mathcal{HZ}_{V_1}\times\Gr(2,n_2)\Big)\cup\Big(\Gr(2,n_1)\times\mathcal{HZ}_{V_2}\Big)
\end{split}
\] is again a hypersurface.
\end{example}
\begin{figure*}[h]

\centering
\begin{tabular}{cc}
\begin{tikzpicture}[scale = 0.7]

\draw [thin, gray] (0,0) grid (4,4);

\draw [<->,thick] (0,4) node (yaxis) [above] {$\alpha_2$}
|- (4,0) node (xaxis) [right] {$\alpha_1$};

\draw[-,dashed] (-0.4,3.4) -- (3.4,-0.4);

\draw[-,dashed] (1,-0.4) -- (1,4.4);

\draw[-,dashed] (-0.4,1) -- (4.4,1);

\coordinate (c1) at (2,1);
\fill[blue] (c1) circle (3pt);
\node at (1.7, 2.5) {$(1,2)$} ;

\coordinate (c2) at (1,2);
\fill[blue] (c2) circle (3pt);
\node at (2.7, 1.5) {$(2,1)$} ;

\end{tikzpicture}
    
     &  
     
     \begin{tikzpicture}[scale = 0.7]

\draw [thin, gray] (0,0) grid (4,4);

\draw [<->,thick] (0,4) node (yaxis) [above] {$\alpha_2$}
|- (4,0) node (xaxis) [right] {$\alpha_1$};

\draw[-,dashed] (-0.4,4.4) -- (4.4,-0.4);

\draw[-,dashed] (1,-0.4) -- (1,4.4);

\draw[-,dashed] (-0.4,1) -- (4.4,1);

\coordinate (c1) at (2,1);
\fill[blue] (c1) circle (3pt);

\coordinate (c2) at (1,2);
\fill[blue] (c2) circle (3pt);

\coordinate (h1) at (1,3);
\fill[yellow] (h1) circle (3pt);
\node at (1.7,3.5) {$(1,3)$};

\coordinate (h2) at (2,2);
\fill[yellow] (h2) circle (3pt);
\node at (2.7, 2.5) {$(2,2)$} ;

\coordinate (h3) at (3,1);
\fill[yellow] (h3) circle (3pt);
\node at (3.7, 1.5) {$(3,1)$} ;

\end{tikzpicture}
\end{tabular}
    \caption{Example~\ref{ex:hurwitz}. The blue points are the formats for which the associated variety $\CZ_{\tilde{V},\alpha}$ is a hypersurface. The yellow points are the ones with $\mathcal{HZ}_{\tilde{V},\alpha}$ a hypersurface and cut out by the inequalities and the equality $\alpha_1+\alpha_2=4,\alpha_1\geq 1,\alpha_2\geq 1$.}
    \label{fig:hurwitz}
\end{figure*}
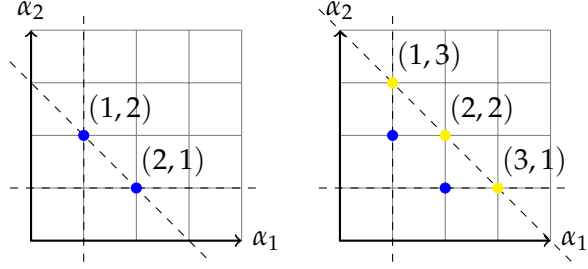

Similar to the case of associated varieties of multiprojective varieties, we can classify all formats $\alpha$ where $\mathcal{HZ}_{V,\alpha}$ is a hypersurface.

\begin{theorem}
\label{thm:multiHurwitz}
       Let $V\subset\PP^{\bm n}$ be an irreducible multiprojecive variety of dimension $r$ and let $\alpha\leq\bm{n}$ be a format with $|\alpha|=\codim V$. If $\alpha\not\in\supp(V)$, then the following are equivalent.
       \begin{enumerate}
           \item $\mathcal{HZ}_{V,\alpha}\subset\prod_{i=1}^l \Gr(\alpha_i,n_i)$ is a hypersurface.
           \item For every $\varnothing\neq I\subset [l]$ we have \[
           \sum_{i\in I} n_i-\alpha_i \leq \dim\pi_I(V)+1.
           \]
           \item There exists a format $\gamma\leq\alpha$ such that $|\gamma|=\codim V-1$ and $\CZ_{V,\gamma}$ is a hypersurface.
       \end{enumerate}
       If $\alpha\in\supp(V)$, then $\mathcal{HZ}_{V,\alpha}$ is a hypersurface if and only if $\mdeg(V,\alpha)\neq 1$.
\end{theorem}
\begin{proof}
The proof is similar to the proof of \cite[Proposition~3.1]{multigraded}, and it takes about two pages.
$(1\Rightarrow 2)$ For $\alpha\not\in\supp(V)$, we have \[
\mathcal{HZ}_{V,\alpha} = \{L\in\prod_{i=1}^l \Gr(\alpha_i,n_i)\mid L\cap V\neq\varnothing\}.
\] Consider the incidence variety \[
Z = \{(p,L)\in V\times\prod_{i=1}^l\Gr(\alpha_i,n_i)\mid p\in L\}
\] and the double filtration \[
V \,\leftarrow\, Z \,\rightarrow\, \mathcal{HZ}_{V,\alpha}.
\] Note that both projections are surjective. For a point $p\in V$, the fiber over $p$ is given by \[
\{L\in\prod_{i=1}^l\Gr(\alpha_i,n_i)\mid p\in L\}
\] which is itself a product of Grassmannians and has dimension $\sum_{i=1}^l \alpha_i(n_i-\alpha_i)$. Hence, the incidence variety $Z$ has dimension \[
r+\sum_{i=1}^l \alpha_i(n_i-\alpha_i)=\sum_{i=1}^l(\alpha_i+1)(n_i-\alpha_i) = \dim(\prod_{i=1}^l \Gr(\alpha_i,n_i)).
\] In particular, $\mathcal{HZ}_{V,\alpha}$ is a hypersurface if and only if the generic fiber over a linear subspace $L\in\mathcal{HZ}_{V,\alpha}$ has dimension $1$. Note that for $L\in\prod_{i=1}^l\Gr(\alpha_i,n_i)$, the fiber over $L$ equals $L\cap V$.

Assume that the inequalities in ($2$) are not satisfied, i.e., there exists $\varnothing\neq I\subset [l]$ such that $\sum_{i\in I}n_i-\alpha_i>\dim\pi_I(V)+1$. Then \[
r-\dim\pi_I(V) > r-(\sum_{i\in I} n_i-\alpha_i)+1 = 1+\sum_{i\not\in I} n_i-\alpha_i
\] holds. In this case we will prove that $\dim(V\cap L)$ is at least $2$ for a generic linear subspace $L\in\mathcal{HZ}_{V,\alpha}$. For such $L$, denote by $L_I$ the product $\prod_{i\in I} L_i$ and similarly $L_{I^c}=\prod_{i\not\in I}L_i$. Then, we have \[
V\cap L = \big(V\cap\pi_I^{-1}(L_I)\big)\cap\pi_{I^c}^{-1}(L_{I^c}).
\] Since $L\in\mathcal{HZ}_{V,\alpha}$, we have $V\cap\pi_I^{-1}(L_I)\neq\varnothing$. This implies that $\pi_I(V)\cap L_I\neq\varnothing$ so $V\cap\pi_I^{-1}(L_I)$ has dimension at least the generic fiber dimension of $\pi_I$, i.e., \[
\dim(V\cap\pi_I^{-1}(L_I))\geq r-\dim\pi_I(V)>1+\sum_{i\not\in I}n_i-\alpha_i.
\] Since $\codim\pi_{I^c}^{-1}(L_{I^c})=\sum_{i\not\in I}n_i-\alpha_i$, we get $\dim(V\cap L)=\dim(V\cap\pi_I^{-1}(L_I)\cap\pi_{I^c}^{-1}(L_{I^c}))>1$. Therefore, $\mathcal{HZ}_{V,\alpha}$ is not a hypersurface in this case.

$(2\Rightarrow 3)$ Assume that for each $\varnothing\neq I\subset [l]$ the required inequality holds. We will show that there exists an index $j\in [l]$ such that $j\in I$ implies $\sum_{i\in I}n_i-\alpha_i\leq\dim\pi_I(V)$. Then, setting $\gamma=\alpha-\bm{e}_j = (\alpha_1,\dots,\alpha_{j-1},\alpha_j-1,\alpha_{j+1},\dots,\alpha_l)$, for any $I\subset [l]$ we have \[
\sum_{i\in I}n_i-\gamma_i = \begin{cases}
\sum_{i\in I} n_i-\alpha_i & \text{ if }j\not\in I\\
(\sum_{i\in I} n_i-\alpha_i) + 1 & \text{ if }j\in I
\end{cases}\;\leq\, \dim\pi_I(V)+1.
\] By Proposition~\ref{prop:formats}, we deduce that $\CZ_{V,\gamma}$ is a hypersurface. 

To prove the existence of the index $j$, we set \[
S_\alpha \coloneqq \{I\subset [l]\mid \sum_{i\in I}n_i-\alpha_i=\dim\pi_I(V)+1\}.
\] Note that $S_\alpha\neq\varnothing$ since $\alpha\not\in\supp(V)$. We recall that the function $I\mapsto\dim\pi_I(V)$ has the property that for $I,J\subset [l]$, \[
\dim\pi_I(V)+\dim\pi_J(V)\geq \dim\pi_{I\cup J}(V)+\dim\pi_{I\cap J}(V)
\] holds (Remark~\ref{rem:submodular}). Then, for $I,J\in S_\alpha$, we have \[
\begin{split}
    \sum_{i\in I\cup J} n_i-\alpha_i &= \sum_{i\in I} n_i-\alpha_i+\sum_{i\in J}n_i-\alpha_i-\sum_{i\in I\cap J} n_i-\alpha_i\\
    &\geq \dim\pi_I(V)+1+\dim\pi_J(V)+1-\dim\pi_{I\cap J}(V)-1\\
    &\geq \dim\pi_{I\cup J}(V)+1.
\end{split}
\] Thus, $I\cup J\in S_\alpha$. On the other hand, for $I=[l]$ we have \[
\sum_{i\in [l]} n_i-\alpha_i = |\bm{n}|-|\bm{\alpha}|=r = \dim(V),
\] which implies that $[l]\not\in S_\alpha$. Since $S_\alpha$ is closed under taking unions and $[l]\not\in S_\alpha$, we deduce that there is an index $j\in [l]$ such that $I\in S_\alpha$ implies $j\not\in I$. The result follows. 


$(3\Rightarrow 1)$ By permuting the indices if necessary
 we may assume without loss of generality that \[
 \alpha = \gamma+\bm{e}_1=(\gamma_1+1,\gamma_2,\dots,\gamma_l).
 \]
 
 Consider the incidence variety \[
 Z = \{(L,\tilde{L})\in\CZ_{V,\gamma}\times\mathcal{HZ}_{V,\alpha}\mid L\subset\tilde{L}\} 
 \] and the double filtration \[
 \CZ_{V,\gamma}\, \leftarrow\, Z \,\rightarrow\, \mathcal{HZ}_{V,\alpha}.
 \] Note that both projections are surjective. For $L\in\CZ_{V,\gamma}$, the fiber over $L$ equals \[
 \{\tilde{L}\in\prod_{i=1}^l \Gr(\alpha_i,n_i)\mid L_1\subset\tilde{L}_1, L_2=\tilde{L}_2,\dots,L_l=\tilde{L}_l\}
 \] which is isomorphic to the Grassmannian $\Gr(\alpha_1-\gamma_1-1,n_1-\gamma_1-1)\cong\PP^{n_1-\gamma_1-1}$. Thus, $Z$ has dimension \[
 \dim\CZ_{V,\alpha}+n_1-\gamma_1-1 = \sum_{i=1}^l (\gamma_i+1)(n_i-\gamma_i)+n_1-\gamma_1-2.
 \] Similarly, for $\tilde{L}\in\mathcal{HZ}_{V,\alpha}$, the fiber over $\tilde{L}$ is isomorphic to the Grassmannian $\Gr(\gamma_1,\alpha_1)$ with dimension $\alpha_1=\gamma_1+1$. Thus, we have
 \[
 \begin{split}
     \dim\mathcal{HZ}_{V,\alpha} &= \dim Z-\alpha_1=\sum_{i=1}^l (\gamma_i+1)(n_i-\gamma_i) + n_1-2\gamma_1-3\\\
     &=\sum_{i=1}^l (\alpha_i+1)(n_i-\alpha_i) - 1=\dim(\prod_{i=1}^l \Gr(\alpha_i,n_i)) -1.
 \end{split}
 \] 
\end{proof}

\begin{example}
The condition that $\mdeg(V,\alpha)\neq 1$ is necessary and reminisces the condition $\deg(V)\neq 1$ in the projective case. For a linear subspace $L=L_1\times L_2\times\dots\times L_l$ of format $\alpha$ we have \[
\mathcal{HZ}_{L,\bm{n}-\alpha} = \Big\{ \, K_1\times K_2\times\dots\times K_l\in\prod_{i=1}^l \Gr(n_i-\alpha_i,n_i)\; \mid \; \exists j\in [l],\, \dim (K_j\cap L_j)\geq 1\, \Big\}.
\] This variety is simply the union of higher associated varieties of each factor $L_i$, each having codimension $2$. Hence, $\mathcal{HZ}_{L,\bm{n}-\alpha}$ has itself codimension $2$.

For a slightly more interesting example let $V\subset\PP^{n}\times\PP^{n}$ be the multiprojective variety given as the zero set of the standard bilinear product, $V=\VV(\langle x,y\rangle)$. That is, $(x,y)\in V$ iff $\sum_{i=0}^{n} x_iy_i=0$.  By direct computation we can see that for a linear subspace $L=l\times\{[v]\}$ of format $(1,0)$ we have \[
V\cap L = (l\cap [v^\perp])\times\{[v]\},
\] where $[v^\perp]=\{[w]\in\PP^n\mid \langle v,w\rangle=0\}$. For a generic $L$, the projective line $l$ intersects the hyperplane $[v^\perp]$ at a single point. Hence, $\mdeg(V,(1,0))=1$. On the other hand, $V\cap L$ is infinite if and only if $l\subset [v^\perp]$. Consider the variety \[
\mathcal{HZ}_{V,(1,0)} = \{l\times\{[v]\}\in\Gr(1,n)\times\Gr(0,n)\mid l\subset [v^\perp]\}
\] and the projection $\pi:\mathcal{HZ}_{V,(1,0)}\rightarrow\PP^n$ onto the second coordinate. Then $\pi$ is surjective and the fiber over a point $[v]$ is isomorphic to $\Gr(1,[v^\perp])$ which has dimension $2n-4$. Hence, \[
\dim\mathcal{HZ}_{V,(1,0)} = n+2n-4=3n-4
\] holds. Since $\dim \Gr(1,n)\times\Gr(0,n)=3n-2$, we deduce that $\mathcal{HZ}_{V,(1,0)}$ is not a hypersurface.
\end{example}

\section{Combinatorial structure of the support of a multiprojective variety} \label{sec:fun}
In this section, we outline an interesting connection between multiprojective varieties and the polymatroid theory. 

Recall from Theorem~\ref{thm:supp} that the support $\supp(V)$ of an irreducible multiprojective variety $V\subset\prod_{i=1}^l \PP^{n_i}$ are cut out by the inequalities of the form \begin{equation}
\label{eq:polym}
	\sum_{i\in I} (n_i -\beta_i) \leq \dim \pi_I(V),
\end{equation}
 where $I$ runs over all possible subsets of $[l]$. It is immediate from the inequalities that $\supp(V)$, i.e. the set of $\beta$ that satisfies (\ref{eq:polym}), is the set of lattice points of a rational polytope. Remark~\ref{rem:submodular} shows that the function $I\mapsto \dim\pi_I(V)$ has special properties, i.e., it is \textit{submodular}, which further makes $\supp(V)$ a polymatroid. See Definition~\ref{def:polymatroid} below for the definition of a polymatroid.

As demonstrated in Figure~\ref{fig:chow}, the set of formats $\alpha$ for which $\CZ_{V,\alpha}$ is a hypersurface equals to the set of lattice points that lie ``below'' $\supp(V)$. Furthermore , the set of formats for which the higher associated variety is a hypersurface are the lattice points that lie ``above'' the set of non-degenerate Chow formats as in Figure~\ref{fig:hurwitz}. The exact meaning of lying below and above were given in Proposition~\ref{prop:formats} and Theorem~\ref{thm:multiHurwitz}. In this section, we will translate these results to the language of polymatroid theory. More specifically, we will show that the operations of taking points below or above a polymatroid correspond to \textit{truncation} and \textit{elongation} of polymatroids, respectively, which also implies that the set of non-degenerate formats of Chow/Hurwitz forms are polymatroids as well.

\begin{definition}
\label{def:polymatroid}
Let $l\in\NN$ and $\delta:2^{[l]}\rightarrow\NN$. Then, $\delta$ is called \textbf{submodular} if \begin{enumerate}
\item $\delta(\varnothing)=0$,
\item for $I\subset J\subset [l],\;\;\delta(I)\,\leq\,\delta(J)\,$, and,
\item for $I,J\subset [l],\;\;\delta(I)+\delta(J)\,\geq\,\delta(I\cup J)+\delta(I\cap J)$.
\end{enumerate}
For a submodular function $\delta$, the set \[
\mathcal{P}(\delta)\coloneqq\{\alpha\in\NN^l\,\mid\, \sum_{i\in I}\alpha_i\leq\delta(I),\,I\subset[l]\}
\] is called the (discrete) \textbf{polymatroid} associated to $\delta$.
\end{definition}

In the case that $\delta(I)\leq |I|$ for $I\subset [l]$, $\mathcal{P}(\delta)$ is called a \textit{matroid} with the \textit{rank function} $\delta$. Note that in this case $\mathcal{P}(\delta)$ consists of binary vectors (by simply taking $I=\{i\}$ we get $\alpha_i\leq 1$) so we can associate the elements of $\mathcal{P}(\delta)$ to the subsets of $[l]$. In the general case, we can associate the elements of a polymatroid with \textit{multisubsets} of $[l]$, i.e., subsets where the repetition of elements are allowed.

The polymatroids admit properties reminiscent to matroids. We refer to \cite[Section~18]{matroid} for proofs.
\begin{proposition}
\label{prop:polym-equiv}
Assume that $\mathcal{P}$ is a polymatroid.
	\begin{enumerate}
		\item $\mathcal{P}$ is \textbf{downward closed}, i.e., if $\alpha\in\mathcal{P}$ and $\beta\leq\alpha$ then $\beta\in\mathcal{P}$. A vector $\alpha\in\mathcal{P}$ which is not dominated by any other vector in $\mathcal{P}$ is called a \textbf{basis}.
		\item If $\alpha,\beta\in\mathcal{P}$ with $|\beta|<|\alpha|$, then there exists an index $i\in [l]$ such that $\beta+\bm{e}_i\in \mathcal{P}$.
		\item If $\alpha\in\mathcal{P}$ is a basis then $|\alpha|=\delta([l])$.
	\end{enumerate}
\end{proposition}

\begin{remark}
For the remainder of the section, the chief example of a polymatroid to us is the support of a multiprojective variety (and its downward closure). The polymatroids of this form are now called \textit{Chow polymatroids},\footnote{This naming is unfortunate for us because we will later associate a polymatroid to the formats of non-degenerate Chow forms of a variety, which are not themselves Chow polymatroids.} an interesting (and proper) subclass of polymatroids. 
\end{remark}

Now we turn our attention to multiprojective varieties and their supports. To describe the submodular function of $\supp(V)$, we need the following definition.

\begin{lemma}[Dual polymatroid]
\label{lem:dual}
Assume $\mathcal{P}\subset\NN^l$ is a polymatroid associated to the submodular function $\delta:2^{[l]}\rightarrow \NN$. Let $\bm{n}\in\NN^l$ be such that $\forall I\subset [l],\; \delta(I)\leq\sum_{i\in I}n_i$. Then, 
\begin{enumerate}
    \item The function \[
    \delta^*(I)\coloneqq\delta([l]\setminus I)-\delta([l])+\sum_{i\in I}n_i
    \] is submodular.
    \item The set \[
    \mathcal{P}^* \coloneqq \bm{n}-\mathcal{P}= \{\bm{n}-\alpha\mid \alpha\in\mathcal{P}\}
    \] is a polymatroid, associated to the submodular function $\delta^*$, called the \textbf{dual} of $\mathcal{P}$ with respect to $\bm{n}$.
\end{enumerate}
\end{lemma}
\begin{proof}
To prove (1), we simply check the conditions of submodularity. First, we observe that for $I\subset [l]$, \[
\begin{split}
    \delta^*(I) &= \delta([l]\setminus I)-\delta([l])+\sum_{i\in I} n_i\\
    &\geq -\delta(I) +\sum_{i\in I} n_i \\
    &\geq 0
\end{split}
\] where in the second line we used $\delta(I)+\delta([l]\setminus I)\geq \delta([l])$. If $I\subset J$, then \[
\begin{split}
    \delta^*(I)&=\delta([l]\setminus I)-\delta([l])+\sum_{i\in I}n_i\\
    &\leq \delta([l]\setminus J)+\delta(J\setminus I)-\delta([l])+\sum_{i\in J}n_i-\sum_{i\in J\setminus I}n_i\\
    &= \delta([l]\setminus J)-\delta([l])+\sum_{i\in J}n_i+\big(\delta(J\setminus I)-\sum_{i\in J\setminus I}n_i\big)\\
    &\leq \delta^*(J).
\end{split}
\] Lastly, for $I, J\subset [l]$, we have \[
\begin{split}
    \delta^*(I)+\delta^*(J) &= \delta([l]\setminus I)+\delta([l]\setminus J)-2\delta([l])+\sum_{i\in I} n_i+\sum_{i\in J}n_i\\
    &\geq \delta([l]\setminus (I\cup J))+\delta([l]\setminus (I\cap J))-2\delta([l])+\sum_{i\in I\cup J}n_i+\sum_{i\in I\cap J}n_i\\
    &=\delta^*(I\cup J)+\delta^*(I\cap J).
\end{split}
\] Hence, $\delta^*$ is submodular and this finishes the proof of (1). To prove the second claim, we note that a vector $\beta\in\NN^l$ is in $\mathcal{P}^*$ if and only if \[
\forall I\subset [l], \quad \sum_{i\in I}\beta_i \; \leq \; \delta([l]\setminus I)-\delta([l]) + \sum_{i\in I} n_i.
\] Rearranging the inequality we obtain $\sum_{i\in I}(n_i-\beta_i)\leq \delta([l])-\delta([l]\setminus I)$. Note that $\delta([l])-\delta([l]\setminus I)\leq\delta(I)$ holds by the submodularity of $\delta$ so we have $\forall I, \sum_{i\in I}n_i-\beta_i\leq\delta(I)$ which implies that $\bm{n}-\beta\in \mathcal{P}$ and $\bm{n}-\mathcal{P}^*\subset\mathcal{P}$. On the other hand, $(\delta^*)^*=\delta$ so applying the same argument with the roles of $\mathcal{P},\mathcal{P}^*$ swapped, we deduce that $\mathcal{P}^*=\bm{n}-\mathcal{P}$ and this finishes the proof.
\end{proof}
\begin{remark}
Note that by the definiton of $\delta^*$ it is immediate that for any $I\subset [l]$, the inequality $\delta^*(I)\leq\sum_{i\in I}n_i$ holds. This allows us to take dual again, i.e., we can take the dual of $\mathcal{P}^*$ with respect to $\bm{n}$. By Lemma~\ref{lem:dual} (2), it is easy to see that $(\mathcal{P}^*)^*=P$.
\end{remark}

With this definition, Theorem~\ref{thm:supp} translates to the following.
\begin{theorem}[\cite{clz}]
Let $\bm{n}=(n_1,n_2,\dots,n_l)\subset\NN^{l}$ be a vector and let $V\subset\PP^{\bm{n}}$ be an irreducible multiprojective variety. Then, $\supp(V)$ is the set of basis vectors of a polymatroid, which is dual (with respect to $\bm{n}$) to the polymatroid associated to the submodular function $\delta(I)=\dim \pi_I(V)$ where $\pi_I(V)$ denotes the projection of $V$ onto the multiprojective space $\prod_{i\in I}\PP^{n_i}$.
\end{theorem}

Proposition~\ref{prop:formats} and Theorem~\ref{thm:multiHurwitz} describe two interesting combinatorial sets related to $\supp(V)$. In the case of Chow forms, the formats $\alpha$ where $\CZ_{V,\alpha}$ is a hypersurface are all of the form $\beta-e_i$ where $i=1,2,\dots,l$ and $\beta\in\supp(V)$. In the case of Hurwitz forms, the formats are of the form $\alpha+e_i,\; i=1,2,\dots,l$ where $\alpha$ is a format and $\CZ_{V,\alpha}$ is a hypersurface. In particular, both of these sets also have a combinatorial structure that is closely related to $\supp(V)$. Indeed, both of these sets are also polymatroids, obtained by applying structure preserving operations to $\supp(V)$. 

\begin{proposition} 
Let $\mathcal{P}\subset\NN^l$ be a polymatroid associated to the submodular function $\delta$ and $\bm{n}\in\NN^l$ be a vector with $\forall I\subset [l],\, \sum_{i\in I}n_i\geq\delta(I)$. Then, the sets \[
\begin{split}
    \mathcal{P}_T &\coloneqq \{\alpha-\bm{e}_j\mid \alpha\in \mathcal{P},\, j\in [l],\, \alpha_j\geq 1\}\\
    \mathcal{P}^E &\coloneqq \{\alpha+\bm{e}_j\mid \alpha\in\mathcal{P},\, j\in [l],\, \alpha_j<n_j\}
\end{split}
\] are also polymatroids, called the \textbf{truncation} and \textbf{elongation} of $\mathcal{P}$, respectively. Moreover, we have $((\mathcal{P}^*)_T)^*=\mathcal{P}^E$, i.e., the elongation of $\mathcal{P}$ is the dual of the truncation of the dual of $\mathcal{P}$ where the dual is taken with respect to $\bm{n}$.
\end{proposition}
\begin{proof}
  Define a new submodular function $\delta'$ as
  \[
    \forall I\subset [l],\;\delta'(I) \coloneqq\min\{\,\delta(I)\,,\,\delta([l])-1\,\}.
\]
It is straightforward to show that $\delta'(I)$ is submodular. We claim that $\mathcal{P}_T$ is the polymatroid associated to $\delta'$. 
We first prove that
$\mathcal{P}(\delta')\subset\mathcal{P}_T$. Suppose $\beta\in\NN^l$ satisfies $\forall I, \, \sum_{i\in I}\beta_i\leq \delta'(I)$. In particular, we have $|\beta|\leq \delta'([l])=\delta([l])-1$. By Proposition~\ref{prop:polym-equiv} (2) and (3), there exists $\bm{e}_i$ such that $\beta+\bm{e}_i\in\mathcal{P}$ so $\beta\in\mathcal{P}_T$.

For the reverse inclusion, we assume that
$\beta=\alpha-\bm{e}_j\in \mathcal{P}_T$ where $\alpha\in\mathcal{P}$ and $\alpha_j\geq 1$. We need to prove that $\beta$ satisfies the inequalities $\sum_{I}\beta_i\leq\delta'(I)$ for $I\subset [l]$. There are two cases. If $j\in I$ then \[
\sum_{i\in I}\beta_i=\sum_{i\in I}\alpha_i -1\leq \delta(I)-1\leq \delta'(I).
\] 
Hence we may assume that $j\not\in I$. To reach a contradiction, assume that $\sum_{i\in I}\beta_i=\sum_{i\in I}\alpha_i>\delta'(I)$. Then we must have $\delta(I)=\delta([l])$ and $\delta'(I)=\delta([l])-1$. This implies that $\sum_{i\in I}\alpha_i=\delta([l])$. Since $|\alpha|\leq\delta([l])$ also holds, we deduce that $\sum_{i\not\in I}\alpha_i=0$. This contradicts $\alpha_j\geq 1$ and finishes the proof that $\mathcal{P}_T=\mathcal{P}(\delta')$.

The claim $((\mathcal{P}^*)_T)^*=\mathcal{P}^E$ follows from Lemma~\ref{lem:dual} (2) and implies that $\mathcal{P}^{E}$ is a polymatroid since the dual and the truncation operators preserve being a polymatroid. 
\end{proof}
Now we state Proposition~\ref{prop:formats} and Theorem~\ref{thm:multiHurwitz} in the language of polymatroids.
\begin{theorem}
\label{thm:trunc-elong}
Let $V\subset\PP^{\bm n}$ be an irreducible multiprojective variety and assume that $\mdeg(V,\alpha)\neq 1$ for $\alpha\in\supp(V)$. If $\mathcal{P}$ denotes the polymatroid with basis vectors in $\supp(V)$, then \begin{enumerate}
    \item the set of basis vectors of the truncation $\mathcal{P}_T$ of $\mathcal{P}$ equals the set of all formats $\alpha$ such that the associated variety $\CZ_{V,\alpha}$ is a hypersurface, and,
    \item the set of basis vectors of the elongation $\mathcal{P}^E_T$ of $\mathcal{P}_T$ equals the set of all formats $\beta$ such that the higher associated variety $\mathcal{HZ}_{V,\beta}$ is a hypersurface.
\end{enumerate}
\end{theorem}

\section{Acknowledgments}
The authors would like to thank Mat\'{i}as Bender, Kathl\'{e}n Kohn, Jonathan Leake, and Jorge Alberto Olarte for useful discussions. We also thank the anonymous referee for suggesting improvements. A.E. is supported by NSF CCF 2110075, M.L.D. is supported by ERC Grant 787840, E.T. is supported by ANR JCJC GALOP (ANR-17-CE40-0009), the PGMO grant
ALMA, and the PHC GRAPE.

\appendix

\section{Auxiliary results}

The proof of Proposition~\ref{prop:CV-homo-CI}
relies on the following that bounds on the bitsize
of  multivariate polynomial multiplication; see \cite{srur-17}.
\begin{claim}[Polynomial multiplication]
  \label{claim:mul-mpoly}
  The following bounds holds:
  \begin{enumerate}[ label=(\roman*)  ]
  \item Consider two multivariate polynomials, $f_1$ and $f_2$, in $\nu$
    variables of total degrees $\delta$, having bitsize $\tau_1$ and
    $\tau_2$, respectively.
    Then $f = f_1f_2$ is a polynomial in $\nu$ variables,
    of total degree $2 \delta$
    and bitsize $\tau_1 + \tau_2 + 2\, \nu \, \lg(\delta)$.
  \item Using induction, the product of $m$ polynomials in $\nu$ variables,
    $f = \prod_{i=1}^{m}f_i$, each of total degree  $\delta_i$
    and bitsize $\tau_i$, is a polynomial of
    total degree $\sum_{i=1}^{m}{\delta_i}$ and bitsize
    $\sum_{i=1}^{m}\tau_i + 12 \, \nu\, m\, \lg(m) \,
    \lg(\sum_{i=1}^{m}{\delta_i})$.

  \item Let $f$ be a polynomial in $\nu$ variables of total degree $\delta$
    and bitsize $\tau$.
    The $m$-th power of $f$, $f^m$, is a polynomial
    of total degree $m \delta$ and bitsize $m \tau + 12 \nu m \lg(\delta)$.
  \end{enumerate}
\end{claim}


\begin{thebibliography}{99}

\bibitem{determinant}
John Abbott, Manuel Bronstein, and Thom Mulders.
\newblock Fast deterministic computation of determinants of dense matrices.
\newblock In {\em Proceedings of the 1999 International Symposium on Symbolic and Algebraic Computation}, ISSAC '99, pages 197-204, New York, NY, USA, 1999. Association for Computing Machinery.

\bibitem{basuroy}
Saugata Basu, Richard Pollack, and Marie-Fran{\c{c}}oise Coste-Roy.
\newblock {\em Algorithms in Real Algebraic Geometry}, volume~10 of {\em Algorithms and Computation in Mathematics}.
\newblock Springer Berlin / Heidelberg, Berlin, Heidelberg, 2006.

\bibitem{Bernardin-phd}
Laurent Bernardin.
\newblock {\em Factorization of multivariate polynomials over finite fields}.
\newblock PhD thesis, ETH Zurich, 1999.

\bibitem{BraSag-0dim-16}
Cornelius Brand and Michael Sagraloff.
\newblock On the complexity of solving zero-dimensional polynomial systems via projection.
\newblock In {\em Proc. ACM International Symposium on Symbolic and Algebraic Computation (ISSAC)}, pages 151--158, 2016.

\bibitem{caniglia1990compute}
Leandro Caniglia.
\newblock How to compute the chow form of an unmixed polynomial ideal in single exponential time.
\newblock {\em Applicable Algebra in Engineering, Communication and Computing}, 1(1):25--41, 1990.

\bibitem{canny_generalised_1990}
John Canny.
\newblock Generalised characteristic polynomials.
\newblock {\em Journal of Symbolic Computation}, 9(3):241--250, March 1990.

\bibitem{clz}
Federico Castillo, Yairon Cid-Ruiz, Binglin Li, Jonathan Monta{\~{n}}o, and Naizhen Zhang.
\newblock When are multidegrees positive?
\newblock {\em Advances in Mathematics}, 374:107382, 2020.

\bibitem{bitpiss4}
David Castro, Marc Giusti, Joos Heintz, Guillermo Matera, and Luis~Miguel Pardo.
\newblock The hardness of polynomial equation solving.
\newblock {\em Foundations of Computational Mathematics}, 3:347--420, 2003.


\bibitem{CE-subdiv-jacm}
John Canny and Ioannis Emiris.
\newblock A subdivision-based algorithm for the sparse resultant.
\newblock {\em Journal of the ACM (JACM)} 47(3):417--451, 2000.


\bibitem{cayley}
Arthur Cayley.
\newblock On a new analytical representation of curves in space.
\newblock {\em The Quarterly Journal of Pure and Applied Mathematics}, 3:225--236, 1860.

\bibitem{chistov}
Alexander~L. Chistov.
\newblock Polynomial-time computation of the dimension of algebraic varieties in zero-characteristic.
\newblock {\em Journal of Symbolic Computation}, 22(1):1--25, 1996.

\bibitem{cox}
David Cox, John Little, and Donal O'Shea.
\newblock {\em Using algebraic geometry}.
\newblock Graduate texts in mathematics BV000000067 185. Springer, New York u.a., 1998.

\bibitem{dalbec}
John Dalbec.
\newblock {\em Geometry and Combinatorics of Chow Forms}.
\newblock PhD thesis, Cornell University, 1995.

\bibitem{dalbec2}
John Dalbec and Bernd Sturmfels.
\newblock {\em Introduction to Chow Forms}.
\newblock In NeilL. White, editor, {\em Invariant {Methods} in {Discrete} and {Computational} {Geometry}}, pages 37--58. Springer Netherlands, 1995.

\bibitem{DA-spres-02}
Carlos D'Andrea.
\newblock Macaulay style formulas for sparse resultants.
\newblock {\em Transactions of the American Mathematical Society}, 354(7):2595--2629, 2002.

\bibitem{straightening}
J.~D\'{e}sarm\'{e}nien.
\newblock An algorithm for the rota straightening formula.
\newblock {\em Discrete Mathematics}, 30(1):51--68, 1980.

\bibitem{dickenstein2005solving}
Alicia Dickenstein and Ioannis~Z Emiris.
\newblock Solving polynomial equations, volume 14 of.
\newblock {\em Algorithms and Computation in Mathematics}, 2005.

\bibitem{dogan_multivariate}
M.~Levent Dogan, Alperen~A. Erg\"ur, Jake~D. Mundo, and Elias Tsigaridas.
\newblock The multivariate Schwartz--Zippel lemma.
\newblock {\em SIAM Journal on Discrete Mathematics}, 36(2):888--910, 2022.

\bibitem{eisenbud}
David Eisenbud.
\newblock {\em Commutative algebra with a view toward algebraic geometry}.
\newblock Graduate texts in mathematics BV000000067 150. Springer, New York u.a., 1995.

\bibitem{eisenbudharris}
David Eisenbud and Joe Harris.
\newblock {\em 3264 and all that : a second course in algebraic geometry}.
\newblock Cambridge University Press, Cambridge, 2016.


\bibitem{emt-dmm-j}
Ioannis Emiris, Bernard Mourrain, and Elias Tsigaridas.
\newblock Separation bounds for polynomial systems.
\newblock {\em Journal of Symbolic Computation}, 101:128--151, 2020.

\bibitem{gkz}
Israel~M Gelfand, Mikhail Kapranov, and Andrei Zelevinsky.
\newblock {\em Discriminants, Resultants, and Multidimensional Determinants}.
\newblock Modern Birkh{\"{a}}user Classics. Birkh{\"{a}}user Boston, Boston, MA, 1st ed. 1994. edition, 1994.

\bibitem{bitpiss1}
Nardo Gim{\'e}nez and Guillermo Matera.
\newblock On the bit complexity of polynomial system solving.
\newblock {\em Journal of Complexity}, 51:20--67, 2019.

\bibitem{giusti}
Marc Giusti, Joos Heintz, and Contre~Le Conicet.
\newblock La d{\'e}termination des points isol{\'e}s et de la dimension d'une vari{\'e}t{\'e} alg{\'e}brique peut se faire en temps polynomial, 1991.

\bibitem{bitpiss5}
Bruno Grenet, Pascal Koiran, and Natacha Portier.
\newblock On the complexity of the multivariate resultant.
\newblock {\em Journal of Complexity}, 29(2):142--157, 2013.

\bibitem{harris}
Joe Harris.
\newblock {\em Algebraic geometry : a first course}.
\newblock Graduate texts in mathematics BV000000067 133. Springer, New York u.a., corr. 2. printing edition, 1993.

\bibitem{hartshorne}
Robin Hartshorne.
\newblock {\em Algebraic geometry}.
\newblock Graduate texts in mathematics 52. Springer, New York u.a., 1977.

\bibitem{matroid}
Dominic Welsh.
\newblock {\em Matroid Theory.}
\newblock { London u.a.: Acad. Press, 1976. Print.}


\bibitem{hauenstein2021numerical}
Jonathan Hauenstein, Anton Leykin, Jose Rodriguez, and Frank Sottile.
\newblock A numerical toolkit for multiprojective varieties.
\newblock {\em Mathematics of Computation}, 90(327):413--440, 2021.

\bibitem{bitpiss3}
Joos Heintz, Guillermo Matera, and Ariel Waissbein.
\newblock On the time--space complexity of geometric elimination procedures.
\newblock {\em Applicable Algebra in Engineering, Communication and Computing}, 11:239--296, 2001.

\bibitem{bitpiss2}
Joos Heintz and Jacques Morgenstern.
\newblock On the intrinsic complexity of elimination theory.
\newblock {\em Journal of Complexity}, 9(4):471--498, 1993.

\bibitem{vdH:ffsparse}
J.~van~der Hoeven and G.~Lecerf.
\newblock Sparse polynomial interpolation. {E}xploring fast heuristic algorithms over finite fields.
\newblock Technical report, HAL, 2019.

\bibitem{humphreys}
James~E Humphreys.
\newblock {\em Linear algebraic groups}.
\newblock Graduate texts in mathematics BV000000067 21. Springer, New York u.a., 1975.

\bibitem{jeronimo2004computational}
Gabriela Jeronimo, Teresa Krick, Juan Sabia, and Mart{\'\i}n Sombra.
\newblock The computational complexity of the {C}how form.
\newblock {\em Foundations of Computational Mathematics}, 4(1):41--117, 2004.

\bibitem{jeronimo2001computing}
Gabriela Jeronimo, Susana Puddu, and Juan Sabia.
\newblock Computing chow forms and some applications.
\newblock {\em Journal of Algorithms}, 41(1):52--68, 2001.

\bibitem{kaltofen2005complexity}
Erich Kaltofen and Gilles Villard.
\newblock On the complexity of computing determinants.
\newblock {\em Computational complexity}, 13(3):91--130, 2005.

\bibitem{kaltofen1988improved}
Erich Kaltofen and Lakshman Yagati.
\newblock Improved sparse multivariate polynomial interpolation algorithms.
\newblock In {\em International Symposium on Symbolic and Algebraic Computation}, pages 467--474. Springer, 1988.

\bibitem{kapranov-bernd-chow}
Mikhail~M Kapranov, Bernd Sturmfels, and Andrei~V Zelevinsky.
\newblock Chow polytopes and general resultants.
\newblock {\em Duke Mathematical Journal}, 67(1):189--218, 1992.

\bibitem{Klose-binseg}
J{\"u}rgen Klose.
\newblock Binary segmentation for multivariate polynomials.
\newblock {\em Journal of Complexity}, 11(3):330--343, 1995.

\bibitem{kohn_coisotropic_2021}
Kathl\'en Kohn.
\newblock Coisotropic hypersurfaces in {Grassmannians}.
\newblock {\em Journal of Symbolic Computation}, 103:157--177, March 2021.

\bibitem{koiran_randomized_1997}
Pascal~Koiran.
\newblock Randomized and deterministic algorithms for the dimension of algebraic varieties.
\newblock In {\em Proceedings of the 38th Annual Symposium on Foundations of Computer Science}, FOCS '97, page~36, USA, 1997. IEEE Computer Society.

\bibitem{lpr-asympt-17}
Sylvain Lazard, Marc Pouget, and Fabrice Rouillier.
\newblock Bivariate triangular decompositions in the presence of asymptotes.
\newblock {\em Journal of Symbolic Computation}, 82:123--133, 2017.

\bibitem{lipton}
Richard Lipton.
\newblock The curious history of the {S}chwartz-{Z}ippel lemma.
\newblock https://rjlipton.wpcomstaging.com/2009/11/30/the-curious-history-of-the-schwartz-zippel-lemma.

\bibitem{srur-17}
Angelos Mantzaflaris, {\'E}ric Schost, and Elias Tsigaridas.
\newblock Sparse rational univariate representation.
\newblock In {\em Proc. ACM on International Symposium on Symbolic and Algebraic Computation (ISSAC)}, pages 301--308, 2017.

\bibitem{dezeeuw_schwartz}
Hossein~Nassajian Mojarrad, Thang Pham, Claudiu Valculescu, and Frank de~Zeeuw.
\newblock {S}chwartz-{Z}ippel bounds for two-dimensional products.
\newblock {\em arXiv preprint arXiv:1507.08181}, 2015.

\bibitem{multigraded}
Brian Osserman and Matthew Trager.
\newblock Multigraded {Cayley}-{Chow} forms.
\newblock {\em Advances in Mathematics}, 348:583--606, May 2019.

\bibitem{perron}
Oskar Perron.
\newblock Beweis und versch{\"a}rfung eines satzes von Kronecker.
\newblock {\em Mathematische Annalen}, 118(1):441--448, Dec 1941.

\bibitem{ps-cg-12}
Franco Preparata and Michael Shamos.
\newblock Computational geometry: an introduction.
\newblock Springer Science \& Business Media, 2012.

\bibitem{rojas1999solving}
Maurice Rojas.
\newblock \em{Solving degenerate sparse polynomial systems faster}.
\newblock \em{Journal of Symbolic Computation},
28(1-2), 155--186, 1999.



\bibitem{sombra2004height}
Martin Sombra.
\newblock The height of the mixed sparse resultant.
\newblock {\em American Journal of Mathematics}, 126(6):1253--1260, 2004.

\bibitem{Storjohann-det-bit}
Arne Storjohann.
\newblock The shifted number system for fast linear algebra on integer matrices.
\newblock {\em Journal of Complexity}, 21(4):609--650, 2005.

\bibitem{sturmfelsalgorithms}
Bernd Sturmfels.
\newblock {\em Algorithms in invariant theory}.
\newblock Texts and monographs in symbolic computation. Springer, Wien u.a., 2. ed. edition, 2008.

\bibitem{sturmfels_hurwitz_2017}
Bernd Sturmfels.
\newblock The {Hurwitz} form of a projective variety.
\newblock {\em Journal of Symbolic Computation}, 79:186--196, March 2017.

\bibitem{van_der_waerden_hilberts_1929}
Bartel~Leendert Van~der Waerden.
\newblock On {Hilbert}s function, series of composition of ideals and a generalization of the theorem of {Bezout}.
\newblock In {\em Proc. {Roy}. {Acad}. {Amsterdam}}, volume~31, pages 749--770, 1929.

\bibitem{van_der_waerden_varieties_1978}
B.L. Van~der Waerden.
\newblock On varieties in multiple-projective spaces.
\newblock {\em Indagationes Mathematicae (Proceedings)}, 81(1):303--312, January 1978.

\bibitem{Waerdenvander1937}
Bartel~L. Waerden~van der and Wei-Liang Chow.
\newblock Zur algebraischen geometrie. ix. \"{U}ber zugeordnete formen und algebraische systeme von algebraischen mannigfaltigkeiten.
\newblock {\em Mathematische Annalen}, 113:DCXCII--DCCIV, 1937.

\bibitem{yap_fundamental_2000}
Chee-Keng Yap.
\newblock {\em Fundamental problems of algorithmic algebra}.
\newblock Oxford University Press, New York, 2000.


\end{thebibliography}
\end{document}